\DocumentMetadata{}
\documentclass[acmsmall]{acmart}

\acmJournal{TOSN}
\usepackage{graphicx}
\usepackage{balance}
\usepackage{comment}
\usepackage{bm}
\usepackage{color}
\usepackage{subfigure}
\usepackage{amsmath}
\usepackage{tabularx}
\usepackage{graphicx}
\usepackage{booktabs}
\usepackage{amsfonts}
\usepackage{mathtools}
\usepackage{natbib}

\newtheorem{theorem}{Theorem}[section]

\DeclareMathOperator{\diag}{diag}

\newcommand\blfootnote[1]{%
  \begingroup
  \renewcommand\thefootnote{}\footnote{#1}%
  \addtocounter{footnote}{-1}%
  \endgroup
}

\setcopyright{acmlicensed}
\copyrightyear{2024}
\acmYear{2024}
\acmDOI{10.1145/3689635}

\usepackage{blindtext}

\usepackage{xpatch}

\begin{document}

\title{Power-Domain Interference Graph Estimation for Multi-hop BLE Networks}

\settopmatter{printacmref=false}

% \affiliation{%

\author{Haifeng Jia}
\email{ethanjia@sjtu.edu.cn}
\affiliation{%
  \institution{UM-SJTU Joint Insitute, Shanghai Jiao Tong University}
  \streetaddress{800 Dongchuan Road}
  \city{Shanghai}
  \country{China}
}

\author{Yichen Wei}
\email{ethepherein@sjtu.edu.cn}
\affiliation{%
  \institution{UM-SJTU Joint Insitute, Shanghai Jiao Tong University}
  \streetaddress{800 Dongchuan Road}
  \city{Shanghai}
  \country{China}
}

\author{Yibo Pi}
\authornote{Corresponding author.}
\email{yibo.pi@sjtu.edu.cn}
\affiliation{%
  \institution{UM-SJTU Joint Insitute, Shanghai Jiao Tong University}
  \streetaddress{800 Dongchuan Road}
  \city{Shanghai}
  \country{China}
}

\author{Cailian Chen}
\email{cailianchen@sjtu.edu.cn}
\affiliation{%
  \institution{Department of Automation, Shanghai Jiao Tong University}
  \streetaddress{800 Dongchuan Road}
  \city{Shanghai}
  \country{China}
}

% \affiliation{%

% \author{Yibo Pi}
% \affiliation{%

% \affiliation{%

% \affiliation{%


\blfootnote{The paper is an extension of our conference paper accepted at the International Conference on Embedded Wireless Systems and Networks (EWSN), 2023 \cite{ourEWSN}. This work was supported by the National Natural Science Foundation of China under Grant 62201346.}

\begin{abstract}
Traditional wisdom for network management allocates network resources separately for the measurement and communication tasks. Heavy measurement tasks may compete limited resources with communication tasks and significantly degrade overall network performance. It is therefore challenging for the interference graph, deemed as incurring heavy measurement overhead, to be used in practice in wireless networks. To address this challenge in wireless sensor networks, our core insight is to use power as a new dimension for interference graph estimation (IGE) such that IGE can be done simultaneously with the communication tasks using the same frequency-time resources. We propose to marry power-domain IGE with concurrent flooding to achieve simultaneous measurement and communication in BLE networks, where the power linearity prerequisite for power-domain IGE holds naturally true in concurrent flooding. With extensive experiments, we conclude the necessary conditions for the power linearity to hold and analyze several nonlinearity issues of  power related to hardware imperfections. We design and implement network protocols and power control algorithms for IGE in multi-hop BLE networks and conduct experiments to show that the marriage is mutually beneficial for both IGE and concurrent flooding. Furthermore, we demonstrate the potential of IGE in improving channel map convergence and convergecast in BLE networks.
\end{abstract}




% \keywords{datasets, neural networks, gaze detection, text tagging}


\maketitle

\section{Introduction}
  \label{sec:intro}

The interference graph, depicting both the communication and interference channel conditions between network nodes, is a key element for resource management in wireless networks. Given an interference graph, we can estimate the interference of each node to other nodes in the network and allocate network resources (e.g., power, time, and frequency) in a proper manner. Extensive resource allocation schemes have been proposed for wireless sensor networks assuming the availability of the interference graphs~\cite{wsn_interference_3, wsn_interference_1, wsn_interference_2}. However, compared to its usage, the advances in the efficient measurement of interference graphs are still lacking, which greatly hinders the practical use of existing resource allocation schemes in wireless sensor networks. Efficient interference graph estimation (IGE) has great potential in unleashing the power of resource scheduling in wireless sensor networks.

Traditional wisdom for network management is to allocate  network resources (e.g., time and frequency) separately for the measurement and data transmission tasks. In other words, the measurement and data transmission tasks have to compete for the time and frequency resources. As a result, on the same channel, the measurement tasks, if not done efficiently, will significantly reduce the opportunities for data transmission. Unfortunately, under the traditional wisdom for network management, IGE inherently incurs high measurement overhead: a $N$-node network, including $O(N^2)$ total links, consumes at least $N$ slots for measurement, where one node sends a reference signal to all the other nodes for channel estimation in each slot~\cite{ige_traditional_survey, ige_traditional_wisdom}. Such high measurement overhead makes the traditional wisdom not scalable for large networks, not to mention network dynamics requiring frequent updates on the interference graph. To reduce the measurement overhead, another popular approach is to characterize interference by modelling based on the propagation environment, where model parameters are pre-determined and no measurement is needed \cite{channel_modelling}. Similarly, several recent works have
assumed the existence of a direct mapping between node attributes, e.g., geolocation, and resource scheduling decisions, which can be learned from a large amount of network layouts by deep learning~\cite{ige_spatial_dl,ige_embedding}. Without active measurements, these approaches lack the ability to track network dynamics.

In this paper, we propose to integrate measurement tasks into data transmission tasks, contrary to the traditional wisdom of separating them, such that the measurement tasks can be conducted together with the data transmission tasks, without reserving time resources solely for conducting the measurement tasks. To achieve this, our core insight is to decompose the superposition of the received signal strengths from multiple concurrent senders into the signal strengths of individually attenuated transmit powers in the data transmission tasks. Specifically, our approach assumes \emph{the linearity of received power}: the received power of a listener is a linear combination of the channel gains and the transmit powers of the senders. By varying the transmit powers of the senders, we can obtain different received powers of the listener. Then, the channel gains can be obtained by solving a group of equations, if the received powers of the listener and the transmit powers of the senders are all known. Meanwhile, the choice of the transmit powers above should also ensure that the data transmission tasks not be affected due to carrying the measurement tasks.

Although the core insight is intriguing, it is difficult to find a suitable data transmission task for integration, which should satisfy the following three requirements: 1) it needs to be conducted in a well-synchronized network, such that the signal strength measured by the listener in a period can be matched to the set of the senders transmitting in the same period; 2) varying the transmit powers of the senders does not affect, if not improve, the performance of the data transmission tasks; 3) the received powers from concurrent senders are additive for the linearity assumption of received power to be true. Of the three requirements, the last one is the most challenging to satisfy, requiring not only strictly-synchronized concurrent senders for data transmission, but also careful analysis of the hardware features and imperfections of the commercial off-the-shelf (COTS) devices.

Fortunately, we find a perfect match between IGE and concurrent flooding, a technique with increasing popularity in the lower-power wireless networking based on concurrent transmission (CT).  Riding on the tide of concurrent flooding, IGE can be easily introduced into wireless sensor networks. In concurrent flooding, senders rebroadcast received packets in strictly-synchronized time slots, with sub-microsecond synchronization accuracy. This feature of concurrent flooding is a key enabler for the linearity of received power. One side effect of CT is that the carrier frequency offsets (CFOs) between senders may result in strong destructive interference corrupting the received packets, if the individually received powers from the senders are close. This issue can be mitigated with fine-grained power control aided by the interference graph. In other words, concurrent flooding carries the task of IGE and, in return, benefits from better interference management by knowing the interference graph. 
The marriage of IGE and concurrent flooding is mutually beneficial. More importantly, the measured interference graph has the potential to be useful in other network activities, e.g., scheduling. To realize the above benefits, we first verify the assumption of power linearity on the COTS devices (\S\ref{sec:power_linearity_cots}) and analyze the nonlinearity issues (\S\ref{sec:non-linearity}). We propose the sufficient condition for IGE and conduct controlled validation experiments in \S\ref{sec:full_rank_cond} and \S\ref{sec:controlled_ige}. We design network protocols to support simultaneous flooding and IGE (\S\ref{sec:differential_power_adjustment}) and power control algorithms to optimize the transmit powers of nodes based on the estimated interference graph (\S\ref{sec:power_control_alg}). Extensive real-world experiments are conducted to demonstrate the mutual benefits of marrying IGE and concurrent flooding (\S\ref{sec:perf_eval}). IGE is beneficial not only for concurrent flooding where all nodes receive and rebroadcast messages using the same channels, but also for other network activities in BLE multi-hop networks. To demonstrate this, we evaluate how IGE can be leveraged to accelerate channel map convergence and improve the reliability of convergecast (\S\ref{sec:application}).

In summary, we make the following key contributions:\\
\noindent\textbf{(1)} 
We propose to marry the challenging measurement task of interference graph estimation with the data transmission task of concurrent flooding, to conduct interference estimation simultaneously with concurrent flooding and to improve the performance of concurrent flooding with the measured interference graph. \\
\noindent\textbf{(2)}
We validate the power linearity assumption on the COTS devices under BLE 5 PHYs both experimentally and theoretically. We identify, analyze, and mitigate several non-linearity issues of received power related to hardware and time synchronization.\\
\noindent\textbf{(3)}
We design network protocols to integrate interference graph estimation and flooding and optimize concurrent flooding by avoiding strong destructive interference with a power control scheme based on the interference graph. Our approach is implemented atop and compared with BlueFlood, a recent concurrent flooding protocol over Bluetooth.\\
\noindent\textbf{(4)}
We demonstrate two use cases of interference graph estimation for channel map convergence and convergecast in BLE networks to show its potentials.

\section{Background \& Related Work}

\noindent\textbf{Interference Graph Estimation.} 
Traditional wisdom for network management allocates separate network resources for the measurement and data transmission tasks, resulting in the competing relation between the two types of tasks. To control the measurement overhead, cellular networks leverage their capabilities in fine-grained resource allocation and only insert the CSI-IM reference signals on a small portion of subcarriers to measure interference~\cite{cellular_reference_signal_1, cellular_reference_signal_2}. However, low-power COTS devices typically do not have such capabilities and use the whole bandwidth either for measurement, control, or data transmission~\cite{whole_bw_2, whole_bw_1}. This makes the heavy measurement tasks, e.g., IGE, not suitable for wireless sensor networks, despite that extensive works have shown the benefits of knowing the interference graph in wireless sensor networks~\cite{wsn_benefit_ig_2, wsn_benefit_ig_1}. Our work proposes to avoid the heavy measurement overhead of IGE by integrating it into a popular flooding technique in recent years.

\noindent\textbf{Concurrent Flooding.} 
Glossy achieves fast and efficient network flooding in a strictly-synchronized network under IEEE 802.15.4, where senders rebroadcast received packets in each slot~\cite{glossy}. A recent work, BlueFlood, extends concurrent flooding to Bluetooth 5~\cite{nahas:blueflood2019}. Due to the CFOs between concurrent senders, the receiver will see beating patterns with periodic hills (constructive interference) and valleys (destructive interference)~\cite{beating}, which are found under both Bluetooth PHYs and IEEE 802.15.4~\cite{phy_layer_replication}. Beating patterns are strong when the received signal strengths from two concurrent senders are close, causing burst errors due to deep valleys with weak signal strengths~\cite{beating_strength}. 
When the received signal strength from one sender is much larger than the rest, the beating patterns are weak and the capture effect dominates, similar to data transmission with no beating~\cite{rssispy}. Our work extends the use of concurrent flooding for IGE and, further, avoids strong destructive interference in concurrent flooding with power control based on the interference graph.

\noindent\textbf{Power Linearity.} Power linearity is a crucial assumption behind the well-known additive interference model~\cite{additive_interference_model}, but was first found inaccurate for concurrent transmission in~\cite{power_nonlinearity}. Subsequent works~\cite{linearity_exp1,linearity_exp2} reported contradictory experimental results and supported the additive signal strength assumption for concurrent transmission in IEEE 802.15.4. However, prior works are focused on verifying the assumption, without delving into the causes for nonlinearity found in their experiments. We extend the experimental study of power linearity to concurrent transmission in BLE networks and discuss possible causes for nonlinearity issues.  

\noindent\textbf{Interference-aware Scheduling.}
Interference graphs are extensively used for the resource management in wireless networks~\cite{interference_aware_scheduling_2, interference_aware_scheduling_1}. An efficient approach to IGE could greatly unleash the power of resource scheduling for wireless sensor networks. For example, with an interference graph, each sender can choose a proper transmit power to control its interference to others~\cite{interference_graph_power_control}, and more senders can be scheduled to transmit at the same time for better spatial reuse~\cite{interference_graph_spatial_reuse}. Moreover, CT can avoid strong destructive interference with power control. Distributed MIMO  allows the transmitters to adjust their signal phases such that the receiving signals are added constructively. However, it requires strict frequency and time synchronization and involves heavy measurement and communication overheads for BLE nodes to act as distributed MIMO~\cite{RFClock}. In contrast, our approach conducts interference graph estimation in the power domain, which consumes no extra frequency-time resources and is robust to time synchronization error and CFO.

\section{Power-Domain Interference Graph Estimation}

\begin{figure}[t!]
  \centering
  \includegraphics[scale=0.5]{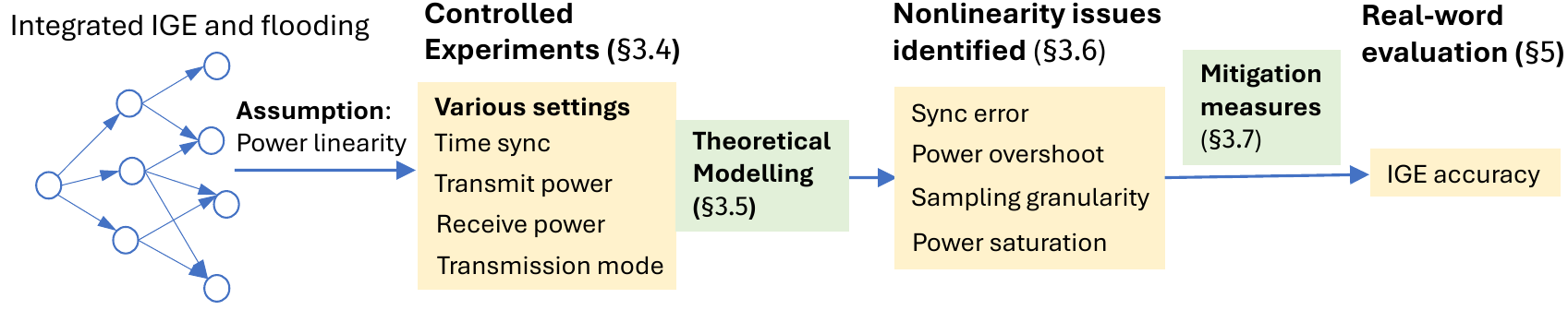}
  \caption{Overview of our experimental study of power-domain IGE}
  \Description{}{}
  \label{fig:overview}
\end{figure}

\subsection{Overview}

Our power-domain IGE builds on the assumption of power linearity, i.e., the individual receive powers from multiple transmitters are additive. We undertake a comprehensive experimental study of power linearity on BLE devices (Figure \ref{fig:overview}). Specifically, we conduct extensive controlled experiments to verify power linearity under various network settings, where the tuning parameters include time synchronization error, transmit power, receive power, and transmission mode. These experiments enable us to 1) identify network settings leading to nonlinearity issues and 2) observe that time synchronization and identical transmit data are the necessary but not sufficient conditions for power linearity. To better understand the sufficient conditions for power linearity, we theoretically model concurrent transmission in BLE networks to gain deeper insights, where factors including CFO, initial phase offsets, and RSSI sampling are considered. Combining those network settings with nonlinearity issues and theoretical modelling, we identify four issues and propose mitigation measures. The resulting power-domain IGE is evaluated in real-world experiments.

\subsection{Core Insight}

\noindent\textbf{An example.}
In Figure \ref{fig:example}, nodes 1 and 2 send a packet to node 3 concurrently in two time slots. In slot 1, nodes 1 and 2 send to node 3 with the transmit powers of 1mW and 2mW, respectively, and node 3 receives  at 3$\mu$W. In slot 2, nodes 1 and 2 change their transmit powers to 1mW and 1mW, respectively, and node 3 receives at 2$\mu$W. Assuming that the received power of node 3 is a linear combination of the transmit powers of nodes 1 and 2, we can have two equations for the two slots and obtain $h_{13} = h_{23} = 0.001$, where $h_{ij}$ is the channel gain from node $i$ to $j$.

\noindent\textbf{Linearity assumption of received power.} This example makes an important assumption about the linearity of received power, making it possible to estimate the channel gains by solving a system of equations. The linearity of received power is a common assumption for interference modelling in wireless networks~\cite{interference_modelling_1, interference_modelling_2}, but past measurement studies have reported different observations~\cite{additivity_hold, additivity_untrue}. We will conduct a comprehensive measurement study of the power linearity and present several non-linearity issues in \S\ref{sec:non-linearity}.

\noindent\textbf{Full-rank constraint.}
This example has a unique solution for the channel gains because the two equations are linearly independent. To generalize this, we want the transmit power matrix, with each row being the transmit powers of nodes in a slot, to be full-rank, such that the channel gains have a unique solution.

\begin{figure}[t!]
  \centering
  \includegraphics[scale=0.35]{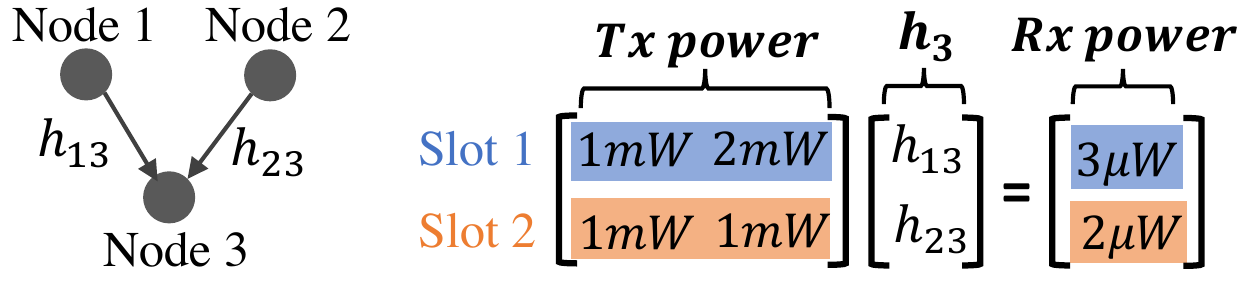}
  \caption{An example of IGE with power control}
  \Description{}{}
  \label{fig:example}
\end{figure}

\subsection{Power Linearity on the COTS Devices}
\label{sec:power_linearity_cots}

\noindent\textbf{Two properties of linearity.}
Under the linearity assumption, the received power of node $i$ from two concurrent senders can be expressed as
\begin{equation}
    p^{rx}_i = h_{1i}p^{tx}_1 + h_{2i}p^{tx}_2,
\end{equation}
where $p^{rx}_i$ and $p^{tx}_i$ are the received and transmit powers of node $i$, respectively. This assumption entails two properties of received power, both of which need to be examined on the COTS devices: 1) \textit{proportionality}, which requires the received power to be proportional to the transmit power, i.e., $p^{rx}_{i\rightarrow j} = h_{ij}p^{tx}_i$, where $p^{rx}_{i\rightarrow j}$ denotes the received power of node $j$ solely from node $i$; 2) \textit{additivity}, which requires the received powers to be additive, i.e., $p^{rx}_i = \sum_j p^{rx}_{j\rightarrow i}$.

\begin{figure}[t!]
  \centering
  \includegraphics[scale=0.45]{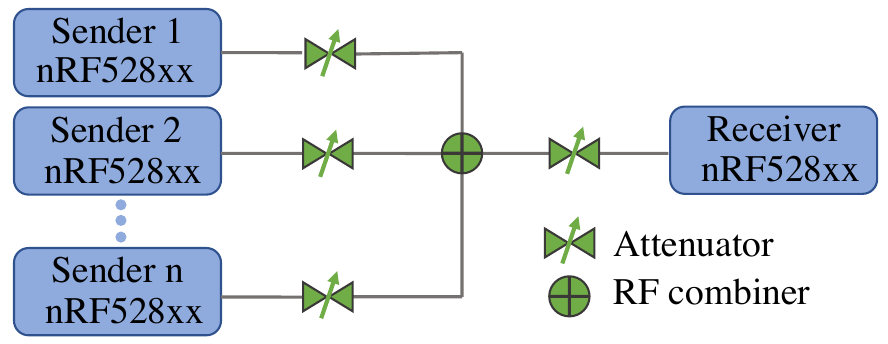}
  \caption{Experimental setup}
  \Description{}{}
  \label{fig:control_exp_setup}
\end{figure}

\noindent\textbf{Experimental setup.} 
We want to examine the two properties of linearity on the COTS devices. To avoid external interference, we conduct our experiments by connecting all nodes with cables. Figure \ref{fig:control_exp_setup} shows the setup with multiple concurrent senders and one receiver, where the attenuators are used to control the received power from each sender at a granularity of 1 dB, and the RF combiner is to mix signals from the concurrent senders. Depending on the specific experiments, different numbers of concurrent senders may be used.
All nodes are equipped with the Nordic nRF52 series SoCs to study Bluetooth 5. Each node sends a packet of 255 bytes to allow sufficient time for the receiver to take RSSI samples. With a sampling rate of $10^6$ samples per second, the receiver collects 1,200 RSSI samples continuously and calculates the average as the received power.

\noindent\textbf{Proportionality between Tx and Rx powers.} 
We expect the received power to be proportional to the transmit power, while hardware imperfections may affect the proportionality. 
To examine this, we need one sender and one receiver for experiments. Figure \ref{fig:proportionality} shows the received powers under different transmit powers and attenuations for nRF52840 SoCs, where the cable and RF combiner together contribute about 3dB attenuation. We can see that the received power increases by almost the same amount in dB as the transmit power does for all attenuations in the \emph{linear region}, where the transmit power is below or equal to 0dBm and the received power is between -90dBm and -20dBm. The linear region closely matches the valid range (-90dBm to -20dBm) of the received power for nRF52 series, but nonlinearity starts to become significant when the received power is below $-90$dBm. We are surprising to find that the transmit powers greater than 0dBm severely deviate from the linear relation with the corresponding received powers. Since there is a sudden jump in deviation for the transmit powers greater than 0dBm, the culprit is likely to be the inaccuracy of the transmit power. We repeated this experiment with six additional pairs of nRF52840 SoCs and found that they all had the same issue.

In the linear region, the transmit power has a strong proportional relation with the received power. We calculate the deviations between the actual and expected received powers, where the expected received power is inferred from the transmit power based on the linear relation. Figure \ref{fig:hardware_imperfection} shows the distribution of  power deviations in the linear region, where the percentages of power deviations below and above 0dB are almost equal. There are about 12\% of power deviations less than -1dB, but all power deviations lie between -2dB and 2dB. This is consistent with the specifications of the nRF52 series, where the received and transmit powers have an accuracy of $\pm 2$dB and $\pm 1$dB in the room temperature, respectively \cite{nrf52832, nrf52840}.

\begin{figure}[t!]
\centering
    \subfigure[Proportionality]
    {\label{fig:proportionality}
		\includegraphics[scale=0.5]{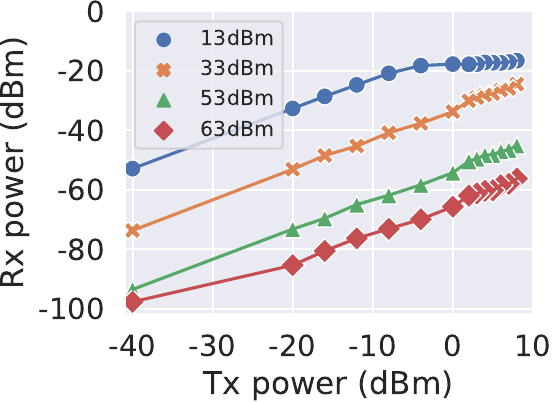}}
    \subfigure[Hardware imperfections]
    {\label{fig:hardware_imperfection}
		\includegraphics[scale=0.5]{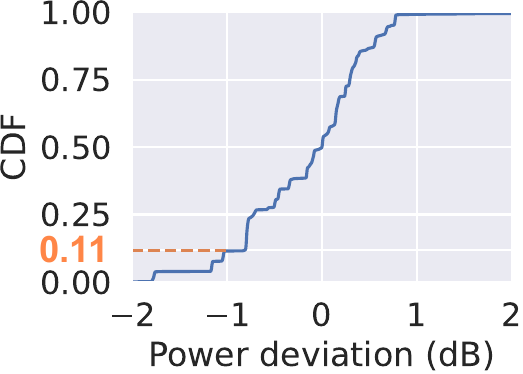}}
    \caption{Proportionality between Tx and Rx powers}
    \Description{}{}
\end{figure}

\noindent\textbf{Additivity of received power.} 
The additivity assumption of received power is commonly used to model interference in wireless networks~\cite{interference_modelling_1, interference_modelling_2}, but it may not be true on the COTS devices. We want to start by checking this assumption under the simple case of two concurrent senders and one receiver. We are concerned with three possible impacting factors: 1) data similarity, i.e., if the concurrent senders transmit the same data, 2) power delta, the difference in received power between the concurrent senders, and 3) time synchronization, i.e., if the concurrent senders are synchronized. The metric to measure power linearity, referred to as the \emph{power ratio}, is defined as the ratio of the actual received power to the sum of individually attenuated powers from the senders. A power ratio close to one indicates strong additivity, and the additivity weakens when the power ratio deviates from one.

To synchronize the two senders, we ask the receiver to first broadcast a short packet to the two senders and then switch to receive mode. After receiving the short packet, the two senders switch to transmit mode and transmit packets concurrently with specified powers to the receiver. In the unsynchronized case, we apply a random delay to the transmission time of one of the senders, where the delay is chosen to be greater than a symbol time to distinguish from the synchronized case with errors. 
We conduct experiments with transmit and received powers in the linear region to exclude the disproportionality issues above. For nRF52 series SoCs, there are only 7 choices for the transmit powers below or equal to 0 dBm. We create granular power deltas between the received powers with adjustable attenuators. In cases where senders need to transmit different packets, we simply ask them to transmit all ones and zeros, respectively.

\begin{figure}[t!]
  \centering
  \includegraphics[scale=0.42]{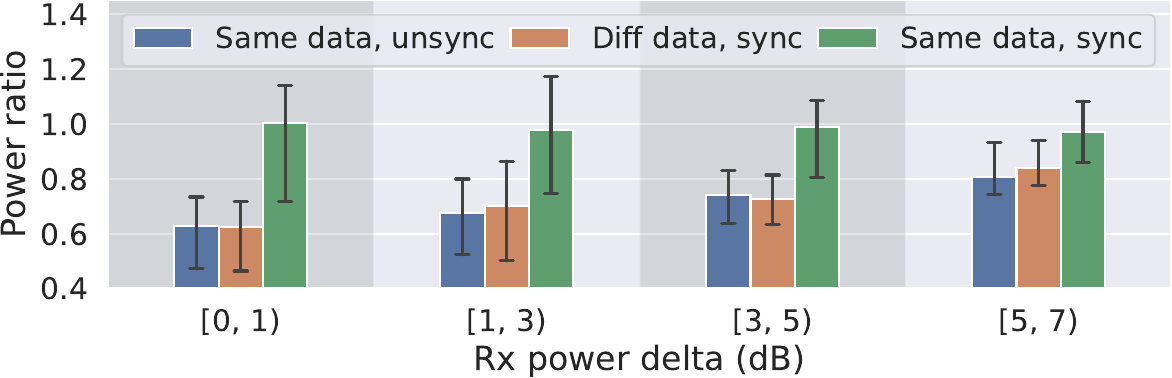}
  \caption{Power ratios under impacting factors}
  \Description{}{}
  \label{fig:conditions_additivity_data}
\end{figure}

% \begin{figure}[t!]
% \centering
%     \subfigure[Conditions for additivity]
%     {\label{fig:cond_additivity}
% 		\includegraphics[scale=0.10]{Images/cond_additivity.jpeg}}
%     \subfigure[Power ratio vs. modes]
%     {\label{fig:power_ratio_vs_mode}
% 		\includegraphics[scale=0.10]{Images/power_ratio_vs_mode.jpeg}}
%     \caption{Additivity of received power}
% \end{figure}

Figure \ref{fig:conditions_additivity_data} shows the distributions of the power ratios under different power deltas, where 
concurrent transmissions under each power delta are repeated 1,000 times with different received powers and with the BLE 1M mode. The bars in the figure show the medians of power ratios and the corresponding error bars show the 5-th and 95-th percentiles of power ratios. We can see that sending unsynchronized same packets or well-synchronized different packets causes all of the power ratios to be far below one regardless of the power delta, indicating weak additivity. In contrast, when sending well-synchronized same packets results in a wide range of power ratios with the median equal to 1 under all power deltas. This implies that time synchronization and data similarity are the necessary but not sufficient conditions for the additivity of received power. In other words, if the packets from concurrent senders are different or if the time is unsynchronized, the superimposed signals have weaker power than the power sum of individually attenuated signals, but sending well-synchronized same packets does not guarantee strong additivity. It can also be seen from Figure \ref{fig:conditions_additivity_data} that 
the range of power ratios shrinks as the power delta increases, and that the power ratios gradually approach to one for sending unsynchronized same packets or well-synchronized different packets. However, this does not mean that the additivity of received power increases with the power delta. Instead, when the power delta is large, the additive power is dominated by the larger received power, which diminishes the impact of non-linearity caused by superposition of the smaller one.

To understand the wide range of power ratios when sending well-synchronized same packets, we first look at the CDF of power ratios, as shown in Figure \ref{fig:cdf_power_ratio}, where about 12\% of power ratios fall beyond the range between 0.9 and 1.1. We will analyze the nonlinearity issues in detail later in \S\ref{sec:non-linearity} and discuss how to identify and remove these issues to reduce the range of power ratios in \S\ref{sec:achieving_linearity}. 
Since the operating range of the received power is critical to power linearity as discussed before, we next relate power ratios with the received power.
Figures \ref{fig:heat_map_ble_high_region} and \ref{fig:heat_map_ble_low_region} show the power ratios under the strong (-40dBm to -20dBm) and weak (-80dBm to -60dBm) received powers, respectively, where the received powers are divided into equal-length bins of 4 dB and the centers of the bins are used as labels for the $x$- and $y$-axes. We can see that the power ratio highly depends on the strength of received powers in the region of strong received powers, where the average power ratio gradually decreases from 1.11 to 0.71 when the received powers both increase from -40dBm to -20dBm along the diagonal. The additivity of received powers is the best when both the received powers are centered around -32dBm or when the received powers differ significantly from each other. The power ratio gradually approaches to 1 as the power delta between received powers increases, because the larger received power dominates, which is consistent with the results in Figure \ref{fig:conditions_additivity_data}. For weak received powers, except for the region near (-68dBm, -68dBm), all other regions have power ratios close to 1. In general, apart from the strong received powers greater than -24dBm, weak and strong received powers demonstrate similar addivitity. The causes for the high power ratios in the regions of strong and weak received powers are different, which will be discussed later in \S\ref{sec:non-linearity}. 

\begin{figure*}[t!]
\centering
    \subfigure[CDF]
    {\label{fig:cdf_power_ratio}
		\includegraphics[scale=0.52]{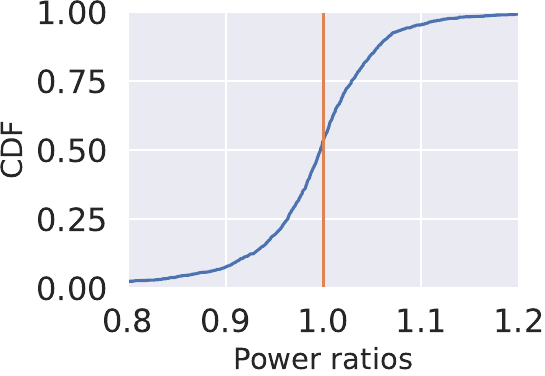}}
    \subfigure[Power ratios under strong received powers]
    {\label{fig:heat_map_ble_high_region}
		\includegraphics[scale=0.52]{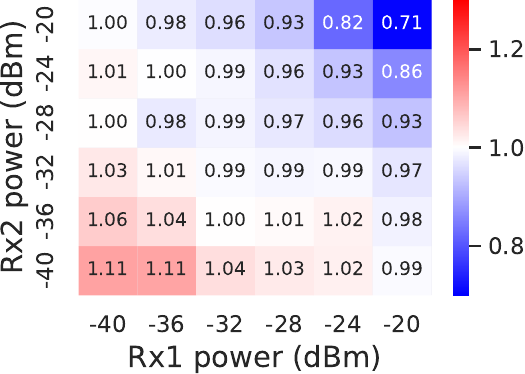}}
    \subfigure[Power ratios under weak received powers]
    {\label{fig:heat_map_ble_low_region}
		\includegraphics[scale=0.52]{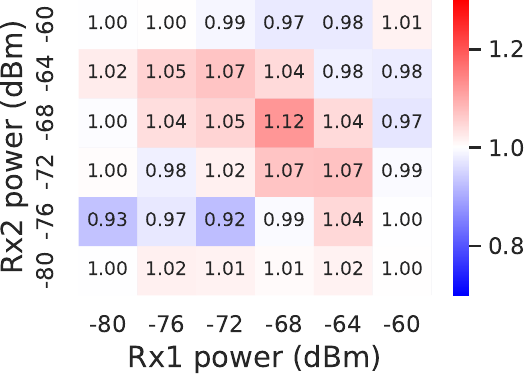}}
    \caption{Conditions for the additivity of received powers}
    \Description{}{}
    \label{fig:conditions_additivity}
\end{figure*}

\begin{figure}[t!]
  \centering
  \includegraphics[scale=0.48]{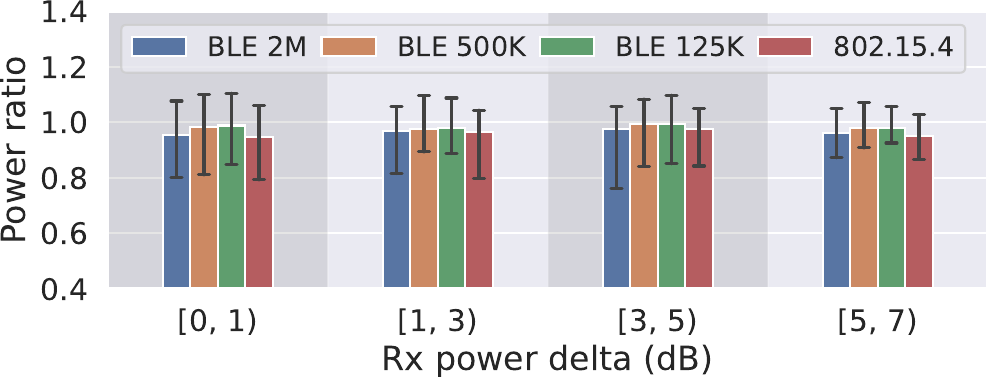}
  \caption{Power ratios under different modes}
  \Description{}{}
  \label{fig:power_ratio_vs_mode}
\end{figure}

We repeat the above experiments under different transmission modes, including the BLE 2M, 500K, 125K modes, and find similar distributions of power ratios as in the BLE 1M mode. Figure \ref{fig:power_ratio_vs_mode} shows that power ratios have a large range including one under all transmission modes and power deltas. This further confirms that sending well-synchronized same packets are necessary but not sufficient to achieve the additivity of received powers. To achieve the additivity of received powers for Bluetooth, the received powers need to be well controlled. Root cause analysis on the nonlinearity issues will be conducted in \S\ref{sec:non-linearity}. By mitigating the nonlinearity issues, the additivity of received powers can be further improved, especially for those with small power deltas and thus having strong influence to each other. Figure~\ref{fig:power_ratio_vs_mode} also shows the distribution of power ratios for 802.15.4, where the average power ratios are close to 1, confirming that power linearity is not limited to Bluetooth.

% the beating cycle is estimated with a logical analyzer [xxx]??. 

\subsection{Theoretical Explanations for Power Linearity} 
\label{sec:theory}

From the experiments above, we observe that data similarity and time synchronization are the necessary but not sufficient conditions for power linearity. We want to explain this by modelling the concurrent transmission in frequency-shift keying (FSK) systems. Let $s_k(t)$ be the modulated transmit waveform from sender $k$, which can be expressed as
\begin{displaymath}
    s_k(t) = \text{Re}(b_k(t)e^{j2\pi f_{tx, k}t}),
\end{displaymath}
where $f_{tx, k}$ is the carrier frequency of the transmitter, and $b_k(t) = a_k(t)e^{j\phi_k(t)}$ is the modulator signal. As Bluetooth employs the Guassian frequency shift keying (GFSK), $a_k(t)$ is constant and $\phi_k(t)$ can be written as
\begin{equation} \label{eq:phi}
    \phi_k(t) = \pi m\int_0^t \sum_{i=0}^{\infty} x_{k, i} g(z - iT_s) dz,
\end{equation}
where $m$ is the modulation index equal to 0.5 for Bluetooth, $x_{k, i}$ is the $i$-th bit transmitted by sender $k$, and $g(t)$ is the Gaussian filter for pulse shaping. If we denote the channel coefficient from sender $k$ to the receiver as $h_k$, the received baseband signal from $N$ senders can be expressed as 
\begin{equation}
    r(t) = \sum_{k=1}^{N} h_{k}b_k(t - \tau_k)e^{j2\pi \Delta f_k t + \Delta \phi^0_k},
\end{equation}
where $\Delta \phi^0_k = \phi^0_{tx,k} - \phi^0_{rx}$ is the initial phase offset between sender $k$ and the receiver, $\Delta f_k = f_{tx, k} - f_{rx}$ is the carrier frequency offset (CFO) between sender $k$ and the receiver, and $\tau_k$ is the time delay from sender $k$ to the receiver.

When sending well-synchronized same packets, we can approximately have $\tau_k = \tau$ and $b_k(t - \tau_k) = b(t - \tau)$ for all $k$'s. Assuming that senders transmit at the same power (i.e., $a_k(t) = a$), the instantaneous power of received signals can be calculated as
\begin{flalign}
\label{eq:beating_pattern}
   |r(t)|^2 = \sum_{k=1}^{N} a^2|h_k|^2 + 2\sum_{k=1}^{N}\sum_{l = k + 1}^N a^2|h_k||h_l|cos(2\pi\Delta f_{k,l}t + \Delta \phi^h_{k,l} + \Delta \phi^0_{k,l}), 
\end{flalign}
where $\Delta f_{k,l} = f_{tx, k} - f_{tx, l}$, $\Delta \phi^h_{k,l} = \phi^h_k - \phi^h_l$ is the difference between the phases of $h_k$ and $h_l$, and $\Delta \phi^0_{k,l} =  \phi^0_{tx,k} - \phi^0_{tx,l}$ is the difference between the initial phases of senders $k$ and $l$. We can see that the power of received signals from concurrent senders has a sinusoidal pattern, typically referred to as the \emph{beating pattern}, with a period of $1/\Delta {f_{k,l}}$, inversely proportional to the difference in carrier frequency between the senders. Equation (\ref{eq:beating_pattern}) indicates that the average received power over a multiple of beating cycles is additive, where the integral of the sinusoidal term is zero. However, the average received power on the COTS devices can only be calculated by averaging the RSSI samples.
Let $T_s$ be the sampling period, $T_b$ be the beating cycle, and $f(t) = cos(2\pi\Delta f_{k,l}t + \Delta \phi_{k,l})$, where $\Delta \phi_{k,l} = \Delta \phi^h_{k,l} + \Delta \phi^0_{k,l}$. The sum of RSSI samples can be written as
\begin{flalign}
\label{eq:discrete_beating}
    \frac{1}{n_s}\sum_{d = 1}^{n_s}|r(T_sd)|T_s = \sum_{k=1}^{N} a^2|h_k|^2 + 2\sum_{k=1}^{N}\sum_{l = k + 1}^N a^2|h_k||h_l|\frac{1}{n_s}\sum_{d=1}^{n_s} f(T_sd)T_s,
\end{flalign}
which is the Riemann sums of Equation (\ref{eq:beating_pattern}). When $T_s \ll 1/\Delta {f_{k,l}}$, the Riemann sum well approximates the integral. In our experiments, the RSSI is sampled every 1 $\mu s$, which is far less than the beating cycles. By taking sufficient samples, we can diminish the impact of the sinusoidal term and approximately have that
\begin{flalign}
    \frac{1}{n_s}\sum_{d = 1}^{n_s}|r(T_sd)|T_s = \sum_{k=1}^{N} a^2|h_k|^2.
\end{flalign}
This explains the power linearity seen in our experiments. More discussions on the non-linearity issues will be discussed later in \S\ref{sec:non-linearity}.

\begin{figure}[t!]
\centering
    \subfigure[An example of beating patterns]
    {\label{fig:beating_pattern_data_similarity}
		\includegraphics[scale=0.5]{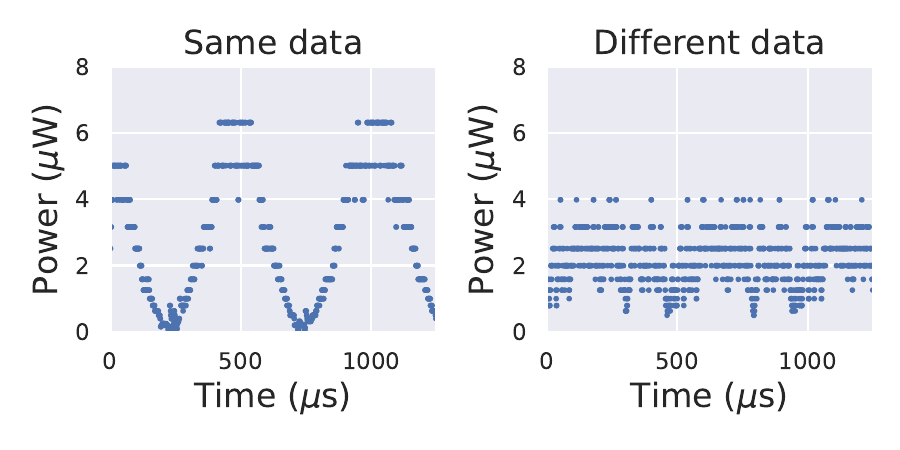}}
    \subfigure[BLE 1M vs. BLE 2M]
    {\label{fig:power_ratio_different_data}
		\includegraphics[scale=0.45]{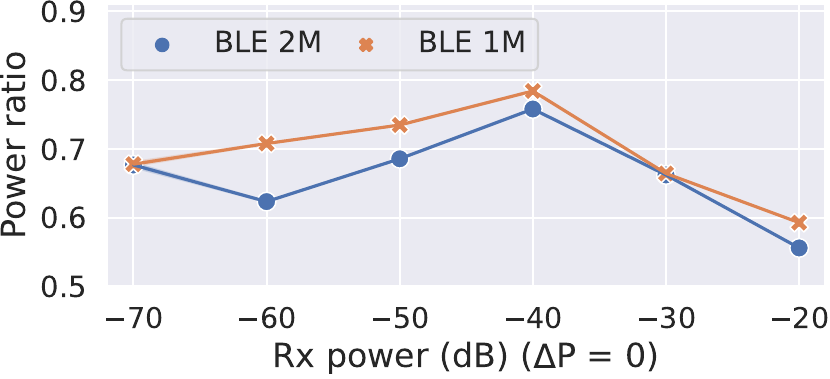}}
    \caption{Power nonlinearity due to ultra-fast beating}
    \Description{}{}
\end{figure}

When sending well-synchronized different packets, we can still approximately have $\tau_k = \tau$, but the modulator signals, $b_k(t-\tau_k)$'s, are different. In Equation (\ref{eq:phi}), $x_{k, i}$ equals to 1 or -1 to represent a bit. When the sender transmits all ones, $\phi_k(t)$ can be simplified to $\frac{\pi}{2T_s}t$, and similarly, when the sender transmits all zeros, $\phi_k(t) = -\frac{\pi}{2T_s}t$. When two senders transmit all ones and zeros, respectively, we can rewrite Equation (\ref{eq:beating_pattern}) by replacing $\Delta f_{k, l}$ with $\Delta f_{k,l}' = f_{tx, k} - f_{tx, l} + \frac{1}{2T_s}$. In our experiments, the typical beating cycles,  $1/\Delta f_{k,l}$, are about hundreds of $\mu s$ when sending well-synchronized same packets. When sending well-synchronized different packets, the extra item, $\frac{1}{2T_s}$, in $\Delta f_{k,l}'$ is $500$kHz for the BLE uncoded 1M and coded modes, making the beating cycle much smaller than $1/\Delta f_{k,l}$. We believe that the ultra-fast beating cycle is the culprit for the power nonlinearity when sending well-synchronized different packets. 

%The derivation above does not consider data whitening, which is often applied to the information bit sequence to make zeros and ones appear as even as possible. In other words, although our information bit sequence includes all ones, the resulting bit sequence after data whitening will contain both zeros and ones. If two sequences differ by a certain amount of bits, the resulting sequences after data whitening will still differ by the same number of bits. Data whitening may affect the

To understand the culprit, we examine the beating patterns between sending the same and different packets. We let two concurrent senders send all ones and zeros, respectively, with the same transmit power and mode. 
Figure \ref{fig:beating_pattern_data_similarity} shows an example of the beating patterns with the BLE 1M mode, where the RSSI sampling period is 1 $\mu s$. We can see a clear beating pattern when sending the same packets, but no beating pattern can be observed when sending different packets. More importantly, half of the RSSI samples are above 2$\mu$W in the beating pattern, but only 24.4\% of samples are above the same threshold when sending different packets. We suspect that this is because the automatic gain control (AGC), used for stabilizing the received power, fails to adapt to the ultra-fast beating. To verify this, we repeat the experiment with the BLE 2M mode, where $\Delta f_{k,l}'$ further increases to about 1MHz. Figure \ref{fig:power_ratio_different_data} shows the average power ratios under different received powers for BLE 1M and 2M. The average power ratios under BLE 2M are always lower than those under BLE 1M. This implies that faster beating results in lower power ratios, which confirms that AGC may over-react to ultra-fast beating and cause lower additive power.

\subsection{Nonlinearity Issues Breakdown}
\label{sec:non-linearity}

\begin{table}
\centering
  \caption{List of nonlinearity issues}
  \label{tab:nonlinearity_breakdown}
    \resizebox{0.5\columnwidth}{!}{%
    \begin{tabular}{c|c|c}
    \toprule
    No. & Issues & Details \\
    \midrule
    %Stable$^{\mathrm{1}}$& Unstable$^{\mathrm{2}}$ & Stable & Unstable & \textbf{measurable} & \\
    1 & Power saturation & Sum of Rx powers exceeds -20 dBm \\
    2 & Overshoot & AGC output overshoots desired level \\
    3 & RSSI sample granularity & Granularity of the RSSI samples is 1dB \\
    4 & Time sync error & Sync error between concurrent senders \\
    \bottomrule
    \end{tabular}
    }
\end{table}

We take an experiment-based approach to discover nonlinearity issues. As in \S\ref{sec:power_linearity_cots}, we
have conducted extensive controlled experiments under various settings of concurrent
transmission and then collected all the network settings with nonlinearity issues. For
each of network settings, we repeated the experiments and collected microsecond-
level time series of RSSI samples for detailed analysis. We discovered four issues in these
RSSI sample series by comparing them with the ideal beating patterns, theoretically
modelled in \S\ref{sec:theory}. Table \ref{tab:nonlinearity_breakdown} lists all of the issues and the detailed descriptions. 

\noindent\textbf{Issue 1: Received power saturation.}
As mentioned before, the maximum nominal received power for the nRF52 series SoCs is -20dBm, beyond which the linear increase of received power suddenly slows down. This indicates that the power linearity will be impaired if the received power saturates.
Suppose that the individually received powers from two concurrent senders are denoted as $P_1$ and $P_2$. From Equation (\ref{eq:beating_pattern}), we know that if $P_1 = P_2 = P$, the beating pattern will oscillate between 0 and 4$P$, where the maxima is four times the individually received power, $P$. This implies that superimposing two individually received powers above -26dBm could saturate the receiver. This agrees with our observations in Figure \ref{fig:heat_map_ble_high_region}, where the received powers above -26dBm result in the power ratios to be less than one. Figure \ref{fig:saturation} shows an example of the received power saturation, where the two individually received powers are both close to -21dBm and the sinusoidal pattern has peaks exceeding the -20dBm limit. We can see that the RSSI samples above 10$\mu$W (i.e., -20dBm) are suppressed, where the peak is 20$\mu$W, which is very close to the desired average, $2P$, and much less than the desired maxima, $4P$.

\begin{figure*}[t!]
\centering
    \subfigure[Received power saturation]
    {\label{fig:saturation}
		\includegraphics[scale=0.40]{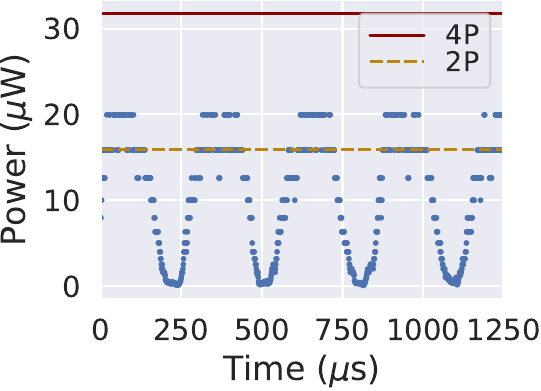}}
    \subfigure[Received power overshoot]
    {\label{fig:overshoot}
		\includegraphics[scale=0.40]{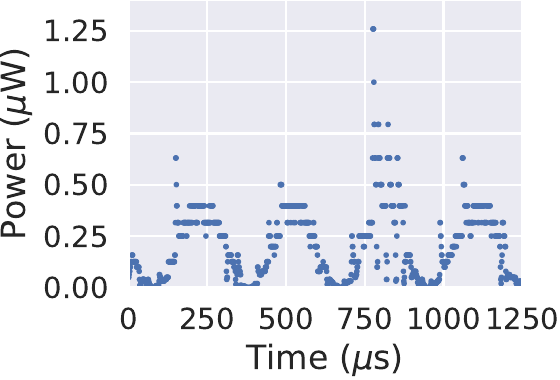}}
    \subfigure[RSSI sample granularity]
    {\label{fig:RSSI_sample_accuracy}
		\includegraphics[scale=0.40]{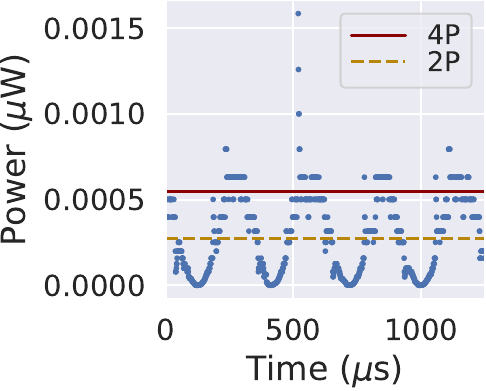}}
    \subfigure[Time sync error]
    {\label{fig:sync_error_nonlinearity}
		\includegraphics[scale=0.40]{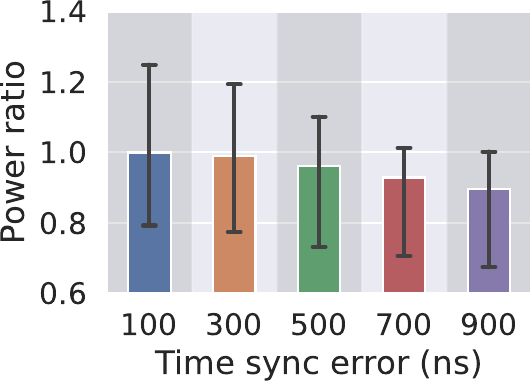}}
    \caption{Nonlinearity issues breakdown}
    \Description{}{}
    \label{fig:nonlinearity_issues_breakdown}
\end{figure*}

\noindent\textbf{Issue 2: Overshoot.} AGC aims to stabilize the output power by dynamically adjusting its power gain. However, when it underestimates the strength of received power, an overlarge power gain may be used, causing the output power overshoot. We observed both short- and long-term overshoots in our experiments. The short-term overshoot typically causes a sudden jump to a certain RSSI sample and takes several microseconds to return to the normal level, while the long-term overshoot may last for a beating cycle and affect all the RSSI samples included. Short-term overshoots could happen anywhere in the beating patterns, but are mostly seen near the peaks in our experiments. In Figure \ref{fig:overshoot}, short-term overshoots occur at the beginning of the first and last peaks, and a long-term overshoot lasts for more than half of the beating cycle. When an overshoot occurs, the original power could be amplified by multiple times. Since the received power is calculated by averaging the RSSI samples, over-amplified RSSI samples will cause the measured received power to be larger than the actual one. This is the major cause for the power ratios to be greater than one when the individually received powers from two concurrent senders are near -68dBm (see Figure \ref{fig:heat_map_ble_low_region}).

\noindent\textbf{Issue 3: RSSI sample granularity.} 
In our experiments, nRF52 series SoCs only allow measuring the RSSI samples at the granularity of 1dB. In most cases, if the received power is rounded up or down to the closet dB evenly, this should not affect the power ratio. However, when the peak RSSI samples take a significant portion of all samples and are all rounded up, the power ratio may become overlarge. Figure \ref{fig:RSSI_sample_accuracy} shows an example of the large power ratio due to the course-grained RSSI sample granularity. About 16\% of samples are at the peaks, which is larger than the desired peak for 15.6\%, where the desired peak is calculated based on Equation (\ref{eq:discrete_beating}).

% When the received power falls below -80dBm, the impact of background noise becomes significant. Figure xxx shows the histogram for background noise in our experiment setup, where the average noise power is -90.8dBm. which is xxx-tenth of the received power. We thus believe that noise is a major cause for the nonlinearity of weak received powers.

\noindent\textbf{Issue 4: Impact of time synchronization errors.} 
In our previous experiments, we synchronize concurrent senders by asking them to switch to transmit mode once receiving the synchronization packet from the receiver. However, due to the hardware imperfections, concurrent senders can hardly be perfectly synchronized and will still experience synchronization errors. In this experiment, we need to accurately measure the time synchronization error in order to study its impact on power linearity. Specifically, suppose that two senders start to transmit packets at $t_1$ and $t_2$, respectively. The time synchronization error is then $|t_1 - t_2|$, which can be measured by comparing the trigger times of the same ratio-registered event between the senders. We use the ADDRESS event in our experiment, which is immediately triggered by the hardware when the access address of a packet is sent out~\cite{nrf52832, nrf52840}.
To avoid CPU delays, we use the Programmable Peripheral Interconnect (PPI) in nRF52 series SoCs to toggle the GPIO pins once the targeted ratio-registered event is triggered, and calculate the difference between the trigger times on the logical analyzer. 
We evaluate the impact of time synchronization errors on power linearity with experiments as follows.

We repeat the experiments of sending well-synchronized same packets above by introducing a random delay below 1000 nanoseconds to one of the two concurrent senders. Figure \ref{fig:sync_error_nonlinearity} shows the distribution of power ratios under different time synchronization errors with the BLE 1M mode, where the $x$-axis labels present the centers of the bins with a width of 200ns. It can be seen that the distribution of power ratios remain relatively stable when the time synchronization error is less than half of the symbol time. However, as the time synchronization error continues to increase, the power ratios begin to decrease. It is interesting that although the senders transmit the same packets, the power linearity still decreases with the synchronization error. This implies that as the synchronization error exceeds half of the symbol time, inter-symbol interference becomes severe. We suspect that this is because data whitening used in Bluetooth, which is often applied to the information bit sequence to make zeros and ones appear as even as possible. In other words, although our information bit sequence comprises all ones, the resulting bit sequence after data whitening will contain both zeros and ones. When the time synchronization error is about a symbol time, it is equivalent to senders transmitting different packets.

\subsection{Achieving Power Linearity}
\label{sec:achieving_linearity}

\begin{figure}[t!]
  \centering
  \includegraphics[scale=0.5]{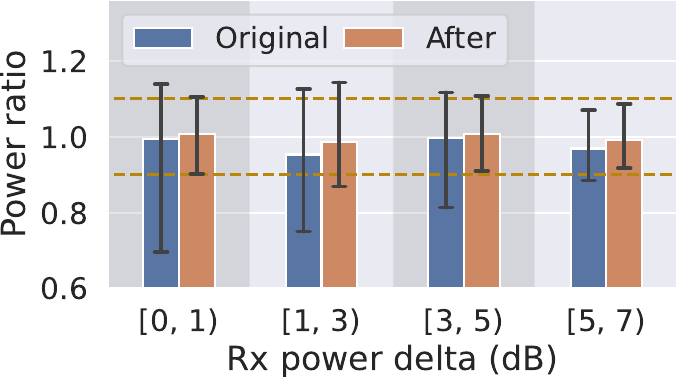}
  \caption{Power ratios before and after mitigation}
  \label{fig:new_power_ratio}
  \Description{}{}
\end{figure}

Based on our observations above, we present methods to mitigate the nonlinearity issues and discuss the power linearity under mutliple concurrent senders.

\noindent\textbf{Mitigating nonlinearity issues.} We want to mitigate the nonlinearity issues from three aspects: 1) limiting the operating range of the transmit and received powers to be within the linear region, i.e., the transmit and received powers are no more than 0dBm and -20dBm, respectively, 2) controlling the time synchronization error to be within one-fourth of the symbol time, and 3) removing abnormal RSSI samples due to overshoot. As shown in Figure \ref{fig:overshoot}, since we have multiple beating cycles, we can remove abnormal RSSI samples by selecting the best beating cycles. This can be done by fitting the RSSI time series with a cosine function and then selecting the two consecutive beating cycle that fits the best, where two consecutive cycles are used to mitigate the error in estimating the beating period. However, nonlinear curve fitting may consume the limited computing resources of the embedded systems, making it impractical for implementation. Considering that
abnormal RSSI samples only cause overlarge power ratios greater than $1.2$ for 0.3\% of concurrent transmissions in our experiment, we employ the first two methods to achieve power linearity in our system implementation (\S\ref{sec:perf_eval}).
Figure \ref{fig:new_power_ratio} shows the ranges of power ratios before and after the methods above are applied. We can see that the most significant improvement is for power ratios with a delta less than 1dB. By removing the outliers, the 5-th and 95-th percentiles of the power ratios are greatly improved.

\begin{figure}[t!]
  \centering
  \includegraphics[scale=0.5]{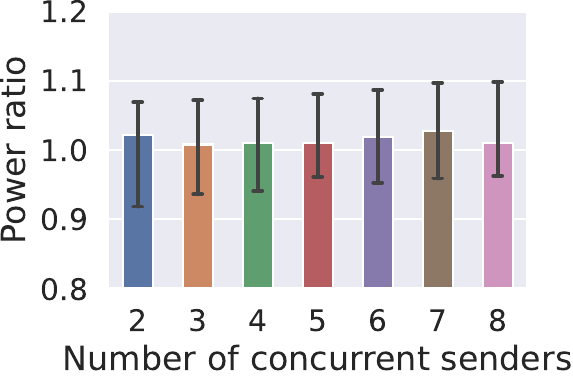}
  \caption{Impact of dense concurrent senders}
  \Description{}{}
  \label{fig:multiple_senders_nonlinearity}
\end{figure}

\noindent\textbf{Dense concurrent senders.}
% {\color{red} randomTxPower: }
Real-world applications may involve multiple concurrent senders. We want to understand if power linearity still holds for the cases of multiple senders. From Equation (\ref{eq:beating_pattern}), we know that the power linearity should still holds if the sampling duration is a multiple of the longest beating cycle. We repeat the experiments above with more concurrent senders, where the receiver continuously collects RSSI samples for 3000$\mu s$. Figure \ref{fig:multiple_senders_nonlinearity} shows the range of power ratios under different numbers of concurrent senders, where the transmit power is selected randomly from [0, -4, -8, -12, -16, -20, -40] dBm. The median of power ratios remain stable as the number of concurrent senders increases, indicating that power linearity holds for dense concurrent senders. Looking at the beating pattern between each sender and the receiver, we find that their beating cycles are between 200$\mu s$ and 400$\mu s$, and that the duration of our RSSI samples includes several their beating cycles.

\subsection{Full-Rank Constraint}
\label{sec:full_rank_cond}

Suppose that there are $n$ senders in a time-slotted multi-hop network. Each node can measure the received signal strength in a time slot it does not transmit data. Let the transmit and received powers of node $i$ at the $j$-th slot be denoted as $p_i^{tx}[j]$ and $p_i^{rx}[j]$, respectively. The channel gain from node $i$ to node $j$ is $h_{ij}$, and the channel gain vector, $\bm{h}_i = [h_{1i},\dots,h_{ni}]^T \in \mathbb{R}^{n}$, includes the channel gains from all senders to node $i$. At a time slot, if the power linearity holds, we can write the received power as a linear combination of the channel gains and the transmit powers, i.e., $p_i^{rx}[j] = \sum_{z=1}^n h_{zi}p_z^{tx}[j]$. We assume that the channel coherence time is much larger than the slot time. By varying the transmit powers of the senders across time slots, we can obtain multiple received powers for node $i$.
Let $\bm{p}_i^{rx} \in \mathbb{R}^{m}$ be the received powers of node $i$ in $m$ time slots, and $\mathbf{P} \in \mathbb{R}^{m\times n}$ be the corresponding transmit power matrix of the senders, where each row of the matrix dictates the transmit powers of senders in a time slot. We can compute $\bm{h}_i$ by solving a constrained linear least square problem for each node, i.e.,
\begin{equation}
    \min_{\bm{h}_i} \Vert \mathbf{P}\bm{h}_i - \bm{p}_i^{rx} \Vert^2, \text{s.t. } h_{ij}^{min} \leq h_{ij} \leq h_{ij}^{max}, \forall i, j.
\end{equation}
In our experiments, we set $h_{ij}^{min}$ and $h_{ij}^{max}$ based on the choice of the transmit powers of nodes such that the received powers fall within the operating range of the Noridc nRF SoCs (see \S\ref{sec:exp_setup}).
If the rank of the transmit power matrix, rank($\mathbf{P}$), is equal to $n$, i.e., $\mathbf{P}$ is \emph{full-rank}, we can obtain a unique solution for $\bm{h}_i$. In order for rank($\mathbf{P}$) to be greater than $n$, the number of time slots, $m$, should be greater than or equal to $n$.

\subsection{IGE: Controlled Experimental Study}
\label{sec:controlled_ige}

\begin{figure*}[t!]
\centering
    \subfigure[Sample size for averaging]
    {\label{fig:num_samples_h}
		\includegraphics[scale=0.37]{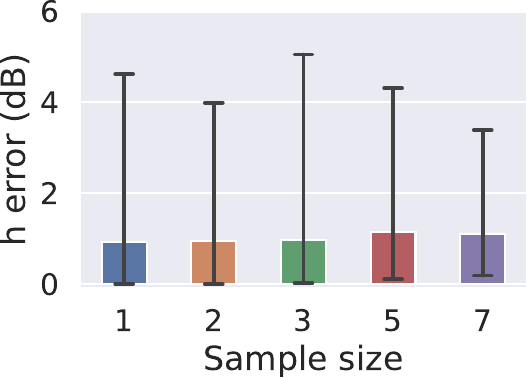}}
    \subfigure[Impact of condition numbers]
    {\label{fig:condition_number}
		\includegraphics[scale=0.37]{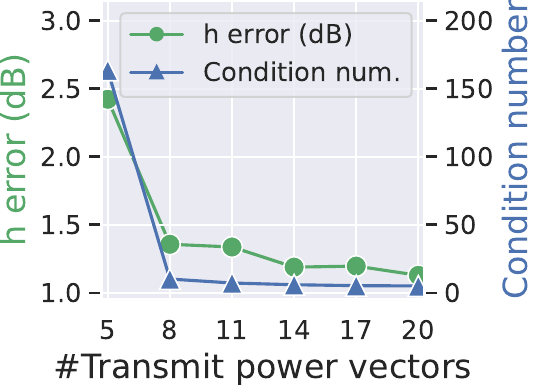}}
    \subfigure[Channel gain accuracy]
    {\label{fig:h_accu}
		\includegraphics[scale=0.37]{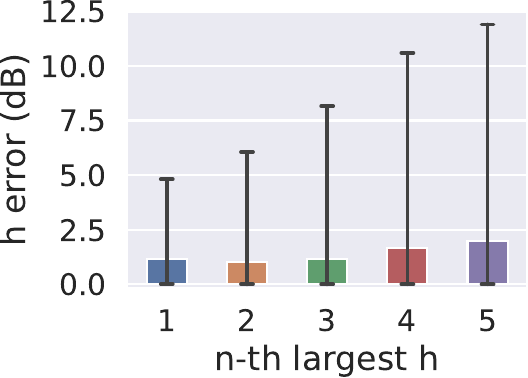}}
    \subfigure[Distribution of channel gains]
    {\label{fig:h_cdf}
		\includegraphics[scale=0.37]{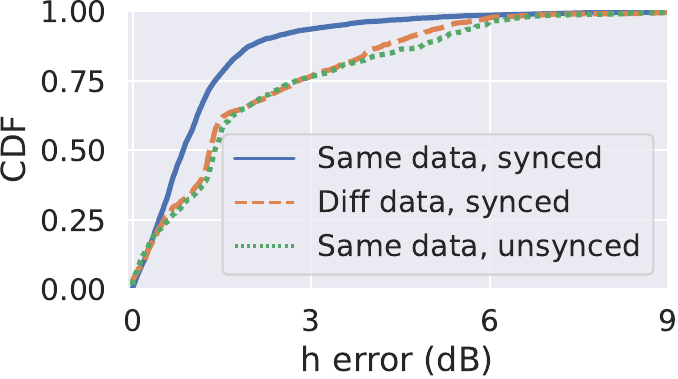}}
    \caption{Feasibility of interference graph estimation on the COTS devices}
    \Description{}{}
    \label{fig:ige}
\end{figure*}

Before discussing how to achieve IGE via CT-based flooding, we first want to validate with experiments that IGE is feasible on the COTS devices. To this end, we use five senders and one receiver in our setup, where the senders are synchronized by a short synchronization packet from the receiver and transmit packets with the BLE 1M mode, as before. To construct a full-rank transmit power matrix, we ask each sender to pick a random transmit power in the linear region for each transmission after the time is synchronized. This process is repeated until we have sufficient measurements for each combination of  transmit powers. Let the $i$-th combination be $\bm{c}_i = [p^{tx}_1[i], \dots, p^{tx}_n[i]]^T$. We need at least $n$ combinations to construct a full-rank matrix for $n$ senders.
We calculate the \emph{channel gain error} as $|\hat{h}_{ij} - h_{ij}|$, where $\hat{h}_{ij}$ is the estimated channel gain for $h_{ij}$ in dB.

With the measurement data above, we first want to understand if averaging multiple measurements for the same combination of transmit powers helps the channel gain estimation. Figure \ref{fig:num_samples_h} shows the channel gain errors when different sample sizes are used to calculate the average transmit powers for each combination. There is no significant improvement for using a larger sample size, which implies that our measured received powers for the same node are stable. We continue this experiment using single samples of received powers to estimate channel gains. Based on the perturbation theory for linear systems, we know that the estimation error of channel gains, $\hat{h}_{ij}$, depends on the condition number of $\mathbf{P}$~\cite{perturbation}. We thus estimate the channel gains using different combinations of transmit powers (or \emph{transmit power vectors}) and obtain a relation between the estimation error and the condition number. We can see from Figure \ref{fig:condition_number} that the average channel gain error decreases together with the condition number. When the total number of transmit power vectors forming the full-rank matrix is more than 11, the condition number starts to decrease at an extremely slow rate. This agrees with the perturbation theory in that a large condition number amplifies the measurement errors and results in a large estimation error. When choosing among candidate transmit power matrices, we prefer the one with a smaller condition number.

The perturbation also suggests that it is more difficult to estimate smaller channel gains. To verify this, we use the adjustable attenuators to create different combinations of channel gains from the five senders to the receiver. In total, 30 combinations of channel gains are created and each is estimated using a full-rank matrix. Based on Figure \ref{fig:condition_number}, we use 11 transmit power vectors for channel gain estimation. Figure \ref{fig:h_accu} shows the estimation error for different magnitudes of channel gains, where the channel gains in each combination is ranked from 1 to 5 in descending order. We can see that the largest channel gains experience the least estimation errors, of which the largest two channel gains can be estimated with the median error about 1dB, the 90-th percentile error less than 3dB, and the 95-th percentile error less than 5dB. If the available transmit powers are more than 3dB apart, the channel gain can then be used to choose a proper transmit power with a high probability. Figure \ref{fig:h_cdf} shows the distribution of the estimation errors for Bluetooth. Sending different or unsynchronized data results in much worse estimation errors.

\section{Integration of IGE and Concurrent Flooding}
\label{sec:integration}

This section presents how to achieve the mutual benefits for both IGE and concurrent flooding from their integration, such that 1) the interference graph can be estimated simultaneously with flooding using the same frequency-time resources and that 2) the performance of concurrent flooding can be improved with the estimated interference graph.

\subsection{Integrating IGE with Flooding: Overview}

\begin{figure}[t!]
  \centering
  \includegraphics[scale=0.45]{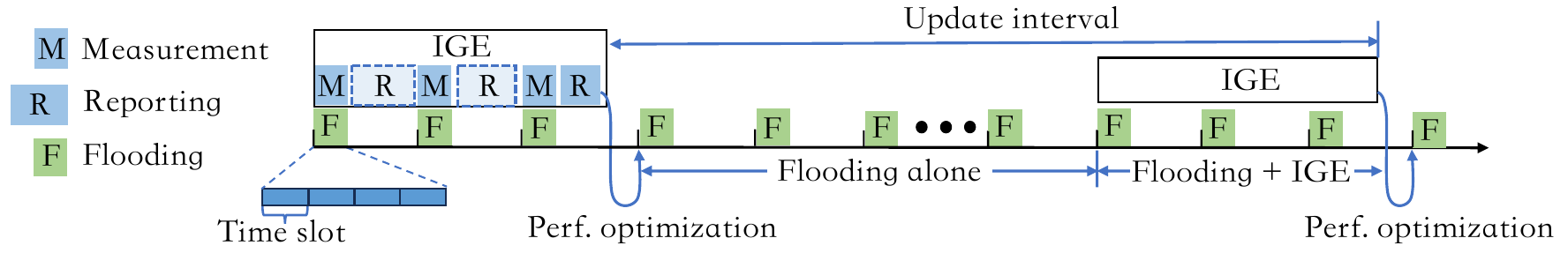}
  \caption{Integrating IGE with flooding}
  \Description{}{}
  \label{fig:IGE_and_flooding}
\end{figure}

We focus on a time-slotted multi-hop network, where an initiator disseminates information to the network via concurrent flooding. There are two major management tasks in the network: IGE and network performance optimization. The IGE is conducted in the power domain to estimate the interference graph, and the estimated interference graph is then used to optimize the transmit powers of nodes for better network performance.

Figure \ref{fig:IGE_and_flooding} shows an example of periodic flooding round after round, each including multiple time slots. The IGE process consists of two subprocesses: measurement and reporting. The measurement process spans several rounds of flooding, each starting with the initiator flooding a packet to the entire network, while the reporting process requires each node to report their local measurements back to the initiator. After collecting local measurements from all other nodes, the initiator updates the estimated interference graph and runs the power control algorithm to generate a new power allocation plan (\S\ref{sec:power_control_alg}), aiming to optimize the performance of flooding.\footnote{The estimated interference graph can be used for other network activities beyond flooding.} Then, the new plan is disseminated throughout the network with flooding. Once receiving the plan, each node updates its transmit power for flooding and retains the transmit power until the next IGE process. To adapt to the changing network conditions, the IGE process is conducted periodically, where the \emph{update interval} is the time between two consecutive IGE processes.

In the measurement process, nodes receive and forward packets as in typical flooding in each round of flooding, except that they need to monitor the signal strengths of received powers when not in transmit mode and to adjust transmit powers as required by the power-domain IGE technique for measurement purposes (\S\ref{sec:differential_power_adjustment}). In the reporting process, as shown in Figure \ref{fig:IGE_and_flooding}, the local measurements of nodes can be opportunistically piggybacked onto the upstream data flows towards the initiator (dashed boxes) and only the last reporting event (solid box) within the IGE process is required. If there is no opportunity for piggybacking before the last flooding event, local measurements collected over several rounds will have to be reported all at once to the initiator. Detailed analysis on the measurement overhead will be discussed in \S\ref{sec:reporting_overhead}. 

% The total duration of the IGE process depends on the number of flooding rounds needed, the inter-round time, and the time taken by the last reporting event. 

% It is noticeable that calculating channel gains requires knowing the received powers of nodes that are locally measured. We can simply assume that the calculation of channel gains is conducted in a centralized way at the initiator and each node needs to report their local measurements back to the initiator during the IGE process. After a plan for power control is created, the initiator disseminates it to all the nodes and each node will adjust its transmit power according to the plan. Since the data collection and dissemination for IGE is not time-critical, it can be piggybacked onto the normal network traffic, e.g., flooding and converagecast, to reduce communication overhead.

\subsection{Performance Optimization for Flooding: Power Control Algorithm}
\label{sec:power_control_alg}

% different hop interference, write along the optimization problem.

Given the estimated interference graph from the IGE proccess, we can coordinate the transmit powers of nodes to optimize the network performance in the flooding-alone period. More specifically, we want to mitigate the strong destructive interference in CT-based flooding, which may corrupt the received packets. It has been shown in~\cite{nahas:blueflood2021} that when two concurrent senders transmit the same data, the packet error rate of the receiver is high due to destructive interference when the individually received powers from the senders are close. We thus want to maximize the power deltas of receivers such that there is a dominant transmitter for each receiver. As the flooded packet is propagated through the network hop by hop, we determine the transmit powers of nodes at the $i$-th hop to maximize the power deltas of the $(i+1)$-hop nodes while considering the interference from nodes at the $(i-1)$-th and earlier hops. We formulate the power allocation problem as follows. 
Suppose that a set of senders, $\mathcal{S}$, are transmitting to a set of nodes, $\mathcal{M}$, with a set of interferers, $\mathcal{I}$, whose transmit powers cannot be controlled. Let $\mathcal{P}$ be the set of available transmit powers. The problem can be written as
\begin{subequations} \label{eq:opt}
\begin{align}
 \underset{\{p^{tx}_i\}_{i\in\mathcal{S}}}{max} \quad & \Delta \\
\textrm{s.t.} \quad & p_{i}^{tx} \in \mathcal{P}, \forall i\in \mathcal{S}, \\ 
&  \log \left( \frac{\hat{h}_{ij}p_{i}^{tx}}{\underset{d \in \mathcal{S} \backslash \{i\} }\sum \hat{h}_{dj}p_d^{tx} + \underset{d \in \mathcal{I}} \sum \hat{h}_{dj}p_d^{tx}}  \right) \geq \Delta, \forall j\in \mathcal{M}, \exists i \in \mathcal{S}, \label{eq:min_power_delta}
\end{align}
\end{subequations}
where $\Delta$ is the power delta between the dominant received power from one sender and the sum of those from the rest and $\hat{h}_{ij}$ is the estimated value for $h_{ij}$. This formulation is a nonlinear integer programming problem, aiming to maximize the minimum power delta of the receivers. Equation (\ref{eq:min_power_delta})
ensures that each receiver has at least one dominant transmitter. Since the number of available transmit powers is limited in the COTS devices, we simple solve this problem with a depth-first branch-and-cut approach, where $\Delta$ is first initialized to negative infinity. During the depth-first search, $\Delta$ is updated if a larger minimum power delta is found on a certain branch, and a branch is cut if it is impossible to find a minimum power delta greater than $\Delta$ by searching more nodes on that branch. 

For a multi-hop network, we compute the allocated powers for nodes at each hop separately, including the last-hop ones. Optimizing the transmit powers of the last-hop nodes aims to improve the packet reception rate of nodes experiencing packet loss at the second-to-last hop. As the power-domain IGE only estimates the downstream channels from the initiator to the last-hop nodes, we assume channel symmetry between the last two hops and re-use the estimated downstream channels to compute the transmit powers of the last-hop nodes. It is noticeable that the power control algorithm needs as input an estimated interference graph. We require the network to conduct point-to-point channel estimation for each link when it starts for the first time. With this estimated interference graph, a full-rank transmit matrix can be computed and, afterwards, the interference graph can be updated periodically as in Figure \ref{fig:IGE_and_flooding}.

\subsection{Measurement Process: Power-Domain IGE}
\label{sec:differential_power_adjustment}

\begin{figure}[t!]
  \centering
  \includegraphics[scale=0.4]{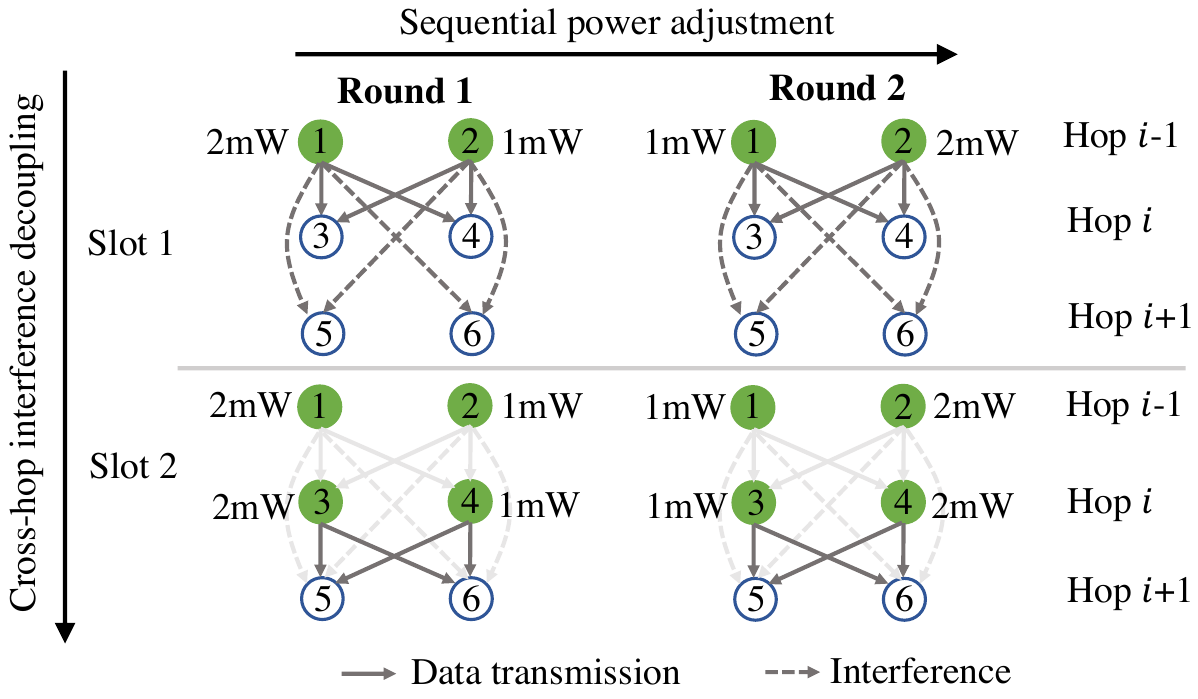}
  \caption{An example of flooding with IGE}
  \Description{}{}
  \label{fig:flooding_w_ige}
\end{figure}

The power-domain IGE measures the interference graph by adjusting the transmit powers of nodes optimized for flooding. To mitigate the impact of IGE on flooding, we propose the \emph{sequential power adjustment} method to adjust the transmit power of each node in turn. To improve the scalability of the power-domain IGE, we further propose the \emph{cross-hop interference decoupling} method to allow nodes at each hop to adjust transmit powers independently. Figure \ref{fig:flooding_w_ige} shows an example of the power-domain IGE in a three-hop network. In each time slot, we use a \emph{communication graph} to describe the relation between senders (solid circles) and receivers (empty circles), where the data transmission and interference links are represented by the solid and dashed arrows, respectively. The sequential power adjustment is conducted across rounds, while the cross-hop interference cancellation is conducted across slots in the same round.

\subsubsection{Sequential Power Adjustment.} 

This method requires that only one node adjusts its transmit power in each round and that nodes adjust their transmit powers in turn across rounds. 
As shown in Figure \ref{fig:flooding_w_ige}, in the first round, node 1 adjusts its transmit power to 2mW, while node 2 remains transmitting at 1mW as in previous rounds. In the second round, node 1 switches back to transmit at 1mW, while node 2 adjusts its transmit power to 2mW. 
Combining the transmit powers of nodes at rounds 1 and 2, we can form a full-rank transmit matrix for nodes at the $(i-1)$-th hop. To improve the scalability of IGE, the power allocation for nodes at different hops is conducted independently, as will soon be discussed in \S\ref{sec:cross-hop-interference}.

% We choose to adjust the transmit powers of nodes in sequence as in the above example for two reasons: 1)  and 2) the choices of available transmit powers on the COTS devices are limited. 

\textbf{Condition for the Full-Rank Transmit Power Matrix}.
We next discuss the conditions to ensure that the transmit power matrix is full rank under the sequential power adjustment. Specifically, let $\boldsymbol{P}_b = \left[\boldsymbol{p}_1^{b}, \dots, \boldsymbol{p}_n^{b}\right]$ and $\boldsymbol{p}_i^b = \left[p_i^b[1],\dots, p_i^b[m]\right]^T$ be the transmit power matrix of nodes at the same hop for flooding, where $n$ is the number of nodes, $p_i^b[j]$ is the transmit power of node $i$ at the $j$-th round, and $m$ is the total number of rounds ($m \geq n$). Since the allocated transmit powers of nodes for flooding remain the same across rounds, we have $p_i^b[j] = p_i^b[k], \forall j\not=k$. Let $\Delta \boldsymbol{P} = \diag(\Delta p_1[1],\dots,\Delta p_n[n])$ be the transmit power adjustment for IGE, indicating that node $i$ adjusts its transmit power at the $i$-th round. The total transmit power matrix is $\boldsymbol{P} = \boldsymbol{P}_b + \Delta \boldsymbol{P}$. We want to carefully choose $\Delta \boldsymbol{P}$ to ensure that $\boldsymbol{P}$ is full rank.

% Channel gains can be estimated by comparing the communication graph in each slot of a round with that of the reference round.
% By comparing the received powers of nodes 3 and 4 in the 1st slots of the reference round and round 1, we can compute the channel gains $h_{13}$ and $h_{14}$. Similarly, we can compute $h_{23}$ and $h_{24}$ by comparing the 2nd slots of the reference round and round 2.

\begin{theorem}\label{thoeorem:full-rank-cond}
To ensure that $\boldsymbol{P}$ is full rank, we need to choose $\Delta p_i[i]$'s such that $\sum\limits_{i=1}^n \frac{p^b_i[1]}{\Delta p_i[i]} \not= 0$.
\end{theorem}

\begin{proof}
Let $\boldsymbol{v} = [p_1^b[i], \dots, p_n^b[i]]^T$ and $\boldsymbol{u} = [1,\dots,1]^T$. We can write $\boldsymbol{P} = \boldsymbol{u}\boldsymbol{v}^T$. Based on the Sherman–Morrison formula, if $\boldsymbol{P}^{-1}$ is invertible and $\boldsymbol{v}^T\Delta \boldsymbol{P}^{-1}\boldsymbol{u} \not= 0$ , $\Delta \boldsymbol{P} + \boldsymbol{u}\boldsymbol{v}^T$ is full rank. In other words, we want to make sure that $\boldsymbol{v}^T\Delta \boldsymbol{P}^{-1}\boldsymbol{u} = \sum\limits_{i=1}^n \frac{p^b_i[i]}{\Delta p_i[i]}\not=0$.
\end{proof}

Our approach needs at least $n$ rounds of flooding to measure the interference graph for a network of $n$ nodes. This approach easily scales to larger networks using more rounds of flooding for measurement, at the cost of prolonging the IGE process. The transmit power matrix $\boldsymbol{P}$ is robust to packet loss and may be rank-deficient only under heavy packet loss, which will be demonstrated with experiments in \S\ref{sec:perf_eval}.

% Theorem \ref{thoeorem:full-rank-cond} provides the sufficient condition for $\boldsymbol{P}$ to be full rank when no packet loss exists. We also want to estimate the probability of $\boldsymbol{P}$ being full rank under packet loss. 

% \begin{theorem}\label{theorem:full-rank-prob}
% Denoting $q$ as the probability of packet loss for each link, we can bound the probability of $\boldsymbol{P}$ being full rank as
% \begin{equation*}
%     \mathbb{P}[\boldsymbol{P} \text{ is full-rank}] > (1 - q)^n + \sum_{i=1}^n\binom{n}{i}\left(q^{i}(1-q)^{n^2-i} + 2(n^2 - n)q^{(i+1)}(1-q)^{n^2-i-1} + 4\sum_{j=2}^n\sum_{k=2}^{j-1}\binom{j-1}{k}q^{(i+k)}(1-q)^{n^2-i-k} \right)
% \end{equation*}
% \end{theorem}

% \begin{proof}
%     Please refer to Appendix \ref{sec:proof-full-rank-prob}.
% \end{proof}

\textbf{Choice of Power Adjustment.}
Theorem \ref{thoeorem:full-rank-cond} provides the condition for the transmit power matrix to be full rank, which ensures the unique solution of channel gains. However, accurate channel gain estimation requires careful selection of $\Delta \boldsymbol{P}$. A larger power adjustment incurs a larger change in the received power, which can combat the measurement errors due to hardware imperfections and thermal noise. Let $\bm{v}$ be the measurement error. We have that
\begin{equation*}
    \boldsymbol{P}\bm{h}_i = (\boldsymbol{P}_b + \Delta \boldsymbol{P}) \bm{h}_i = \bm{p}^{rx}_i + \bm{v} 	\Leftrightarrow \Delta \boldsymbol{P} \bm{h}_i = \Delta \bm{p}^{rx}_i + \bm{v}, 
\end{equation*}
where $\Delta \bm{p}^{rx}_i$ is the change in received power due to power adjustment.
Based on the perturbation theory, we can bound the estimation error of channel gains as
\begin{equation} \label{eq:error_bound}
    \frac{\lVert\hat{\bm{h}}_i-\bm{h}_i\rVert}{\lVert\bm{h}_i\rVert} \leq \kappa(\Delta \boldsymbol{P})\frac{\lVert \bm{v} \rVert}{\lVert \Delta \bm{P}^{rx}_i \rVert},
\end{equation}
where $\hat{\bm{h}}_i$ is the estimated channel gain vector of node $i$ and $\kappa(\Delta \boldsymbol{P})$ is the condition number of $\Delta \boldsymbol{P}$. It can be seen that the estimation error of channel gains can be reduced by decreasing the condition number of $\Delta \boldsymbol{P}$ and increasing the power adjustment. 
However, increasing the power adjustment causes the adjusted power to deviate further from the optimized transmit power for flooding, resulting in a higher packet error rate during the IGE process. This implies a trade-off in network performance between the IGE and the flooding-alone periods. As the flooding-alone period is much larger than the IGE process, we prefer large power adjustment for the accurate estimation of channel gains for the flooding-alone period. This trade-off relation will be demonstrated with experiments in \S\ref{sec:exp_setup}.
As the condition number of a matrix is the ratio of its largest and smallest singular values, we can control the condition number of $\Delta \boldsymbol{P}$ by enforcing all nodes to use the same power adjustment. As implied by Figure \ref{fig:condition_number}, we can also increase the dimension of transmit power matrix to reduce the condition number, but this also requires using a larger number of rounds for the IGE process.

% Specifically, let $p_i^{tx}[j]$ and $\Delta p_i^{tx}[j]$ be the transmit power optimized for flooding and the power. Compared to the reference round, if node $i$ is the only node with a power change in the selected round, based on the linearity property of power, we can have that
% \begin{equation}
%    \Delta p^{rx}_k[j] = h_{ik}\Delta p^{tx}_i[j],
% \end{equation}
% where node $k$ is in listen mode in the $j$-th slot. This indicates that the channel gains from node $i$ to all its downstream nodes can be computed after the reporting process and that 

\subsubsection{Cross-hop Interference Decoupling.} 
\label{sec:cross-hop-interference}

In a multi-hop network, the packet is flooded hop by hop across the network. To further shorten the IGE process, we want to decouple the interference between cross-hop nodes, such that nodes at each hop only need to care about the signals from their direct upstream nodes. After the cross-hop interference is decoupled, we can adjust the transmit powers of nodes at each hop independently to measure the channel gains from these nodes to their downstream nodes. The total number of rounds needed for IGE is thus reduced to $max_i\{n_i\}$, where $n_i$ is the number of nodes at the $i$-th hop.

% In a time-slotted multi-hop network, flooding requires each node to retransmit the received packet for the following $N_{tx}$ consecutive time slots, where $N_{tx}$ is the required number of retransmissions. Like Glossy, we assume that flooding is conducted round by round, each including a fixed number of slots.

Figure \ref{fig:flooding_w_ige} shows how to decouple cross-hop interference by comparing the communication graphs across different slots in the same round. In slot 1, nodes 1 and 2 at the ($i-1$)-th hop begin to rebroadcast a packet at the same time. Nodes 3 and 4 at the $i$-th hop can receive the packet as they fall within the communication range of nodes 1 and 2, while nodes at hop $i+1$ can only sense the weak signal strength from nodes 1 and 2 for being beyond the communication range. In slot 2, nodes at the $i$-th hop have received the packet and begin to rebroadcast it, together with nodes 1 and 2. We assume that nodes 1 and 2 use the same transmit power for both slots 1 and 2. Based on the additivity of received power, since hop-($i+1$) nodes have measured the signal strength from nodes 1 and 2 in slot 1, by subtracting the interference from the total received power in slot 2, they can obtain the strength of signals solely from hop-$i$ nodes. This implies that the cross-hop interference can be decoupled by requiring nodes at the same hop to maintain their transmit powers across time slots in the same round. With interference decoupling, nodes at each hop can estimate the channel gains from themselves to their next hops independently.

\subsection{Reporting Process: Overhead Analysis}
\label{sec:reporting_overhead}

The measurement process above requires us to record the communication graph in each time slot. Specifically, we need to know the status of each node, either in listen, transmit, or idle mode, in each time slot. When a round starts, each node except for the initiator stays in listen mode until it receives the flooded packet and then rebroadcasts it for $N_{tx}$ consecutive rounds. This indicates that only the reception time of the flooded packet needs to be recorded. A round consisting of $n_s$ slots would take each node $\lceil \log n_s \rceil$ bits to represent the reception time. Since we focus on flooding, nodes at the $i$-hop need to take $i$ copies of RSS measurements. Let $b$ be the number of possible RSS values\footnote{The Nordic nRF52 series SoCs measure RSS with a granularity of 1 dB.}, each of which can be represented with $\lceil \log_2 b \rceil$ bits. Let $d$ be the number of total hops. The total communication overhead of our approach for measurement data collection can be estimated as
\begin{equation*}
    O_{our} = \lceil \log n_s \rceil n_r\sum_{i=1}^d n_i + \lceil\log_2 b\rceil n_r \sum_{i=1}^d i \times n_i,
\end{equation*}
where $n_r$ is the total number of rounds of flooding in the IGE process, which needs to be no less than $\max_i\{n_1,\dots, n_{d-1}\}$ to form a full-rank transmit power matrix. The last-hop nodes (at the $d$-th hop) does not affect $n_r$ for having no downstream nodes. Since our approach conducts IGE together with flooding, the time consumption of IGE in our approach is $T_{our} = 0$.

For comparison, we also estimate the communication overhead for the traditional method of IGE, which asks nodes to take turns sending and listening measurement packets. In a multi-hop network, each node only needs to measure signals from upstream nodes. We can thus estimate the total communication overhead of the traditional method as
\begin{equation*}
    O_{other} = \lceil\log_2 b\rceil \sum_{i=1}^d n_i \times \sum_{j=0}^{i-1} n_j,
\end{equation*}
where $n_0 = 1$ presents for the initiator. The time consumption of the traditional method is equal to the number of total nodes, i.e., $T_{other} = \sum_{i=0}^d n_i$ slots.
In a square network where each hop has an equal number of nodes (not including the initiator), we have $n_i = n_r = \sqrt{n}$, where $n$ is the total number of nodes. The communication overhead of our and the traditional methods are $O_{our} = \lceil \log n_s\rceil nd + \lceil\log_2 b\rceil n \frac{d(d+1)}{2}$ and $O_{other} = \lceil\log_2 b\rceil n \frac{d(d+1)}{2}$, respectively. The extra communication overhead of our method is to construct the communication graph, which requires only $\lceil \log n_s \rceil$ bits for each node to record the reception time of the flooded packet. In contrast, our method needs no extra time for measurement, while the traditional method requires $n$ extra time slots.

\section{Performance Evaluation}
\label{sec:perf_eval}

This section presents and evaluates the implementation of our approach based on BlueFlood \cite{nahas:blueflood2021}, a recent lower-power flooding protocol implemented on the Nordic nRF SoCs. % We use Nordic nRF52832 nodes for system implementation and experiments. 

\subsection{Experimental Setup}
\label{sec:exp_setup}

\noindent\textbf{Implementation.} We implement a multi-hop BLE network managed by an initiator. To support the centralized management, we modify Blueflood to make each node capable of adjusting its configurations according to a scheduling plan generated by the initiator. This scheduling plan specifies the transmit powers of nodes for flooding and the power adjustments of nodes for IGE. In the IGE process, each node adjusts its transmit power in the specified time slot and uses the transmit power for flooding for the rest of the round. The number of rounds for IGE is set to the maximum number of nodes at all hops. On each node, we utilize the bit counter to take RSSI samples periodically for every byte received. The received power is estimated as the average of the measured RSSI samples. After the IGE process, nodes sends back the received powers at different time slots to the initiator. We connect the initiator to a laptop for the interference graph estimation and power allocation. 

% {\color{red} Explain the related Figure \ref{fig:deltap1} and \ref{fig:deltap2} }

\begin{figure}[t!]
\centering
    \subfigure[Our testbed]
    {\label{fig:setup}
        \includegraphics[scale=0.03]{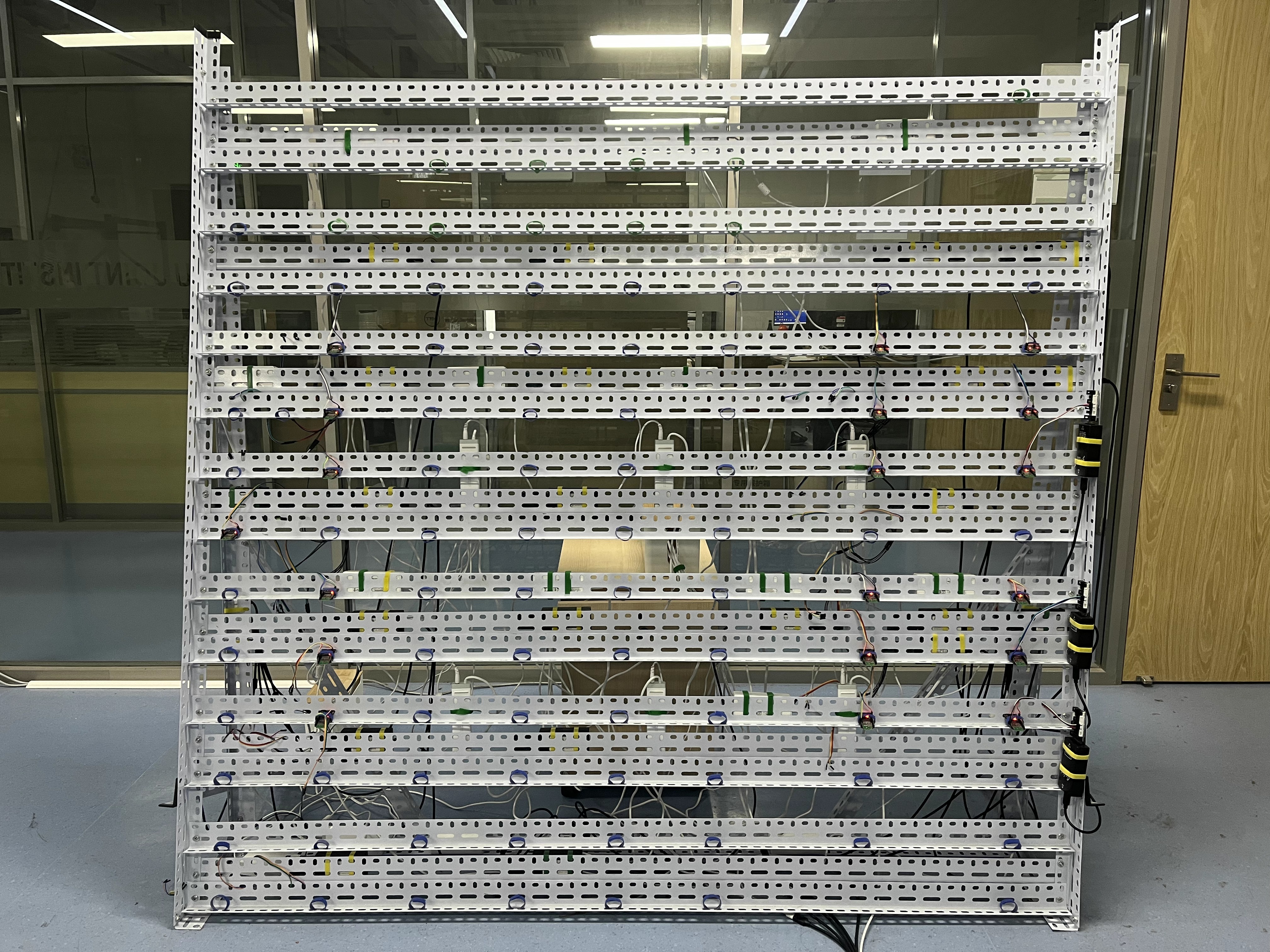}}
    \hspace{3em}
    \subfigure[Network topology]
    {\label{fig:topology}
        \includegraphics[scale=0.45]{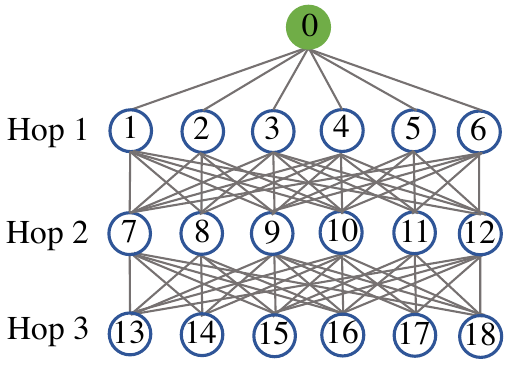}}
    \hspace{3em}
    \subfigure[Channel similarity in our testbed]
    {\label{fig:hdiff}
		\includegraphics[scale=0.45]{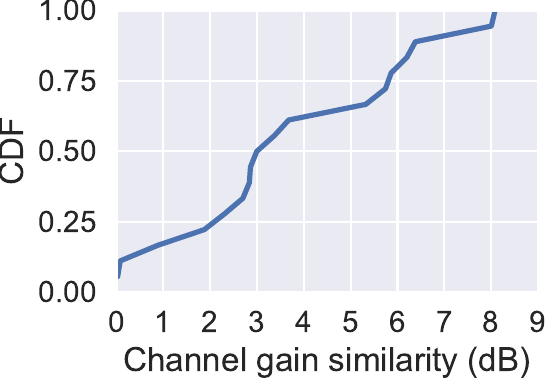}}
    \caption{Testbed}
    \Description{}{}
    \label{fig:test_bed}
\end{figure}

\begin{figure}[t!]
\centering
    \subfigure[Network performance]
    {\label{fig:deltap_net_perf}
        \includegraphics[scale=0.45]{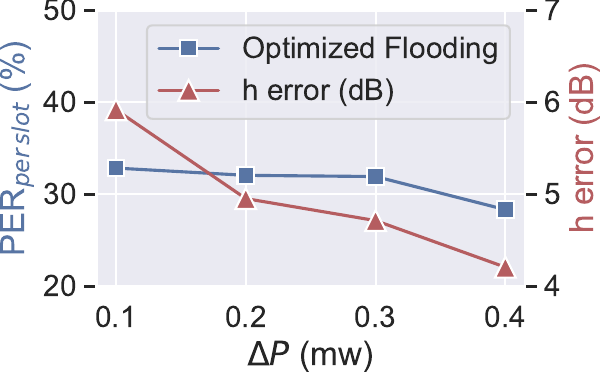}}
    \hspace{3em}
    \subfigure[IGE performance]
    {\label{fig:deltap_ige}
		\includegraphics[scale=0.45]{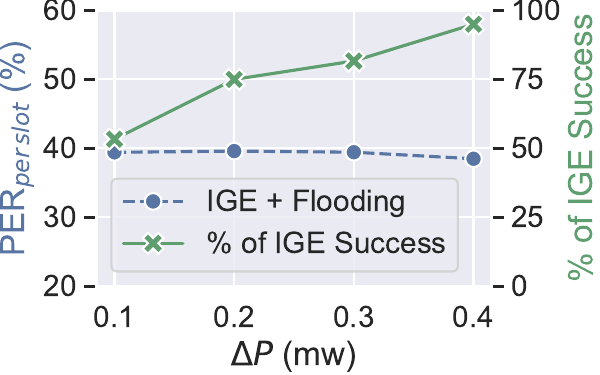}}
    \caption{The impact of power adjustment selection}
    \Description{}{}
    \label{fig:power_adjustment}
\end{figure}

\noindent\textbf{Testbed and Configuration.} We run all the experiments in a 19-node testbed in our office, where the network topology is shown in Figure~\ref{fig:topology}. The deployed network is of three hops with node $0$ being the initiator. Each node can communicate with its neighbor-hop nodes and interfere cross-hop nodes. In our testbed, nodes at the 1st, 2nd, and 3rd hops are spatially aligned, making nodes prone to have similar channel gains to their neighbors and thus susceptible to destructive interference. Figure \ref{fig:hdiff} shows the distribution of the channel gain similarities of nodes, measured by the differences between the two strongest channel gains for all the nodes, where 50\% of nodes have the two strongest channels differing in less than 3dB. Our experiments are conducted on the Bluetooth advertising channel 37. For each node, when receiving a flooded packet, it rebroadcasts the packet for $N_{Tx}$ times. As in BlueFlood, we set the total number of time slots in each round to $N_{tx} + 2 \times \#TotalHops$, where $N_{tx}$ = 4 and $\#TotalHops$ is the total number of hops. In our experiments, we use the transmit powers in the linear region from [-20, -16, -8, -4, 0] dBm. 

% In our experiments, we allocate 6 and 50 rounds for the IGE and flooding-alone periods, respectively, and run the experiment until we have 3360 rounds.

\subsection{Power Adjustment Selection}
We want to select a proper power adjustment for IGE and understand its impact on network performance. As indicated by Equation (\ref{eq:error_bound}), the error bound of IGE can be reduced by enforcing an identical, large power adjustment among all the nodes. 
To verify this, we repeat the experiments of flooding for more than 3,000 rounds under different power adjustments, where each IGE process takes 6 rounds followed by a 50-round flooding-alone period. Due to the limited available transmit powers, we cannot enforce all nodes to adopt the same power adjustment, but require that the power adjustments of nodes be no less than a threshold, denoted as $\Delta P$. Before the start of experiments under each power adjustment, point-to-point channel estimation is conducted between each pair of nodes to obtain the ground-truth channel gains. The channel gain ($h$) estimation error is computed by comparing the estimated channel gain from the power-domain IGE against the ground-truth ones. 
Figure \ref{fig:power_adjustment} shows the impact of power adjustment selection on IGE and network performance. We can see from Figure \ref{fig:deltap_net_perf} that the estimation errors of channel gains decrease significantly as the power adjustment $\Delta P$ increases. As a result, when the estimated channel gains are used to optimize the performance of flooding, the per-slot packet error rate (PER) for flooding in the flooding-alone period also decreases, where the per-slot PER is the ratio of received packets over the number of valid receive slots. Figure \ref{fig:deltap_ige} shows the impact of power adjustment on the IGE process. The stable per-slot PER is a result of two counter-balanced factors: increasing $\Delta P$ results in better channel estimation and power allocation, but also causes higher interference. Although the per-slot PER in the IGE process remains stable under different $\Delta$'s, the large power adjustment changes the packet loss distribution among nodes, with neighbors to the power-adjusted node having higher probability of receiving the flooded packets. This implies that when nodes at one hop adjust their transmit powers in sequence, using large power adjustment increases the likelihood of non-zero diagonal entries of the transmit power matrix at the next hop. We can see that the success probability of IGE increases sharply with $\Delta P$ from 53\% to 95\%. When an IGE process fails, the interference graph from the previous IGE process will be used. The probability of consecutive IGE failures is very low. In our experiments, we found no two consecutive failures of IGE. We use $\Delta P$ = 0.4mW for the following experiments.

\begin{figure}[t!]
\centering
    \subfigure[Overall distribution]
    {\label{fig:overall_dist}
		\includegraphics[scale=0.4]{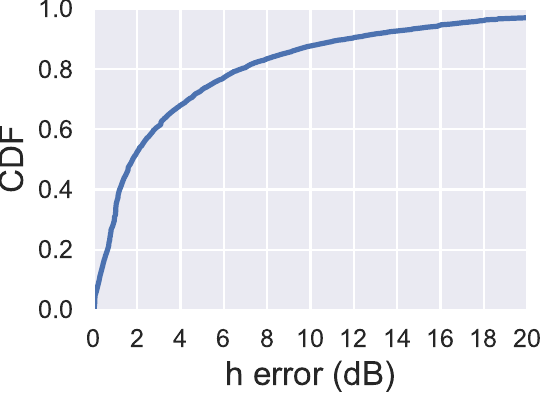}}
    \hspace{3em}
    \subfigure[Estimation errors by channel gain magnitude]
    {\label{fig:est_error_by_channel_gain}
		\includegraphics[scale=0.4]{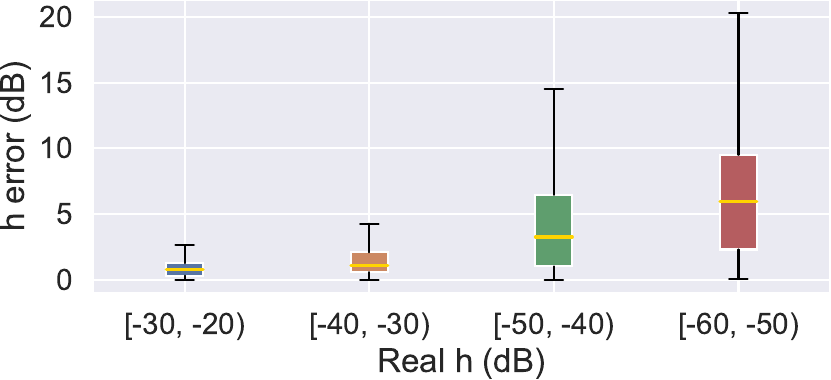}}
    % \hspace{1em}
    % \subfigure[Estimation errors by channel type]
    % {\label{fig:est_error_by_channel_type}
		% \includegraphics[scale=0.4]{Images/herror_vs_comm_intf.pdf}}
    \caption{Accuracy of IGE}
    \Description{}{}
    \label{fig:accuracy_of_IGE}
\end{figure}

\subsection{Accuracy of Power-Domain IGE}

To understand the accuracy of IGE, we compare the channel gains estimated from the power-domain IGE against the ground-truth from the point-to-point channel estimation. The impact of channel variations is mitigated by conducting the power-domain and point-to-point IGEs back to back. We repeat this back-to-back experiments every 10 seconds for 60 times and collect more than 7,000 channel gains. Figure \ref{fig:overall_dist} shows the overall distribution of the estimation errors, where the estimation error is computed as the absolute difference between the estimated and the ground-truth channel gains. We can see that 68\% of channels have an estimation error less than 4dB. Compared with the controlled experiments in \S\ref{sec:controlled_ige}, we find the estimation errors in real-world settings are in general larger due to varying and noisy network conditions. By further grouping channel gains by magnitude (Figure \ref{fig:est_error_by_channel_gain}), we find that the estimation error increases as the channel gain decreases. The power-domain IGE excels at estimating large channel gains with the 75-th percentile error less than 1.8dB for the channel gains greater than -40dB. For weak channels, since they have gains significantly smaller than the dominant channels, they have minimal influence on the resulting received power and thus are hard to be accurately estimated. It is worth noting that the estimation error of a channel in the power-domain IGE depends on whether adding this channel causes significant impact on the received power. Under the same transmit power, the determining factor is the relative magnitude between strong and weak channels, rather than the absolute magnitude. By considering for each node the channel of the strongest individually received power as the dominant channel and all the rest as the interference channels, we find that dominant channels all lie in [-30, -20] dB, which can be accurately estimated with the 75-th and 90-th percentile errors less than 1.8dB and 1dB, respectively. This indicates that the power control algorithm correctly chose the strong channels as the dominant channels.

% the primary source of error for the power control algorithm arises from the interference channels. Although weak channels have a large estimation error, its impact on the total received power is small. 

% We can see that 60\% of the estimated channel gains have an error less than 3dB. For channel gains with an error greater than 3dB, we compute its relative magnitude with the maximum channel gain in the network and find that 70\% of the small channel gains are less than the maximum one for more than 10dB, which significantly mitigates the impact of large estimation errors. 

\begin{figure*}[t!]
\centering
    \subfigure[Per-slot PER]
    {\label{fig:metric1}
            \includegraphics[scale=0.4]{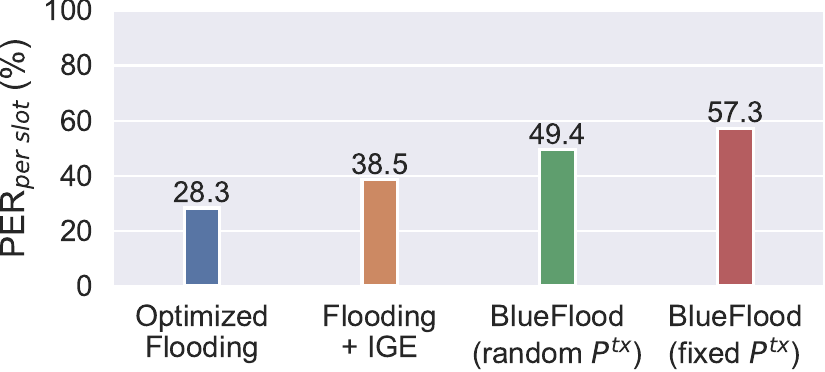}}
    \hspace{3em}
    \subfigure[E2E PER]
    {\label{fig:metric2}
            \includegraphics[scale=0.4]{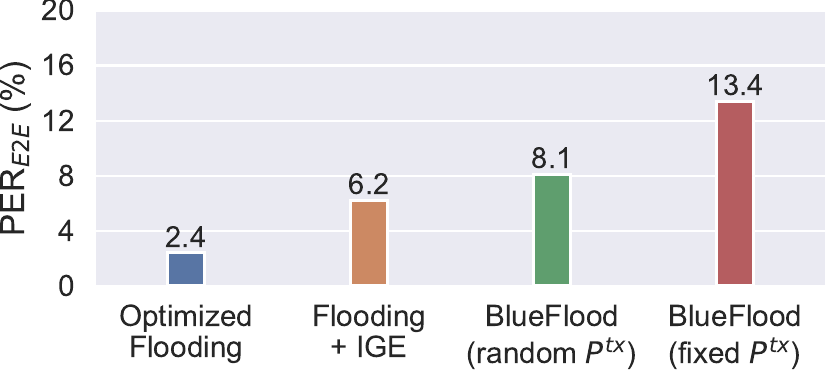}}
    \subfigure[Latency]
    {\label{fig:metric3}
            \includegraphics[scale=0.4]{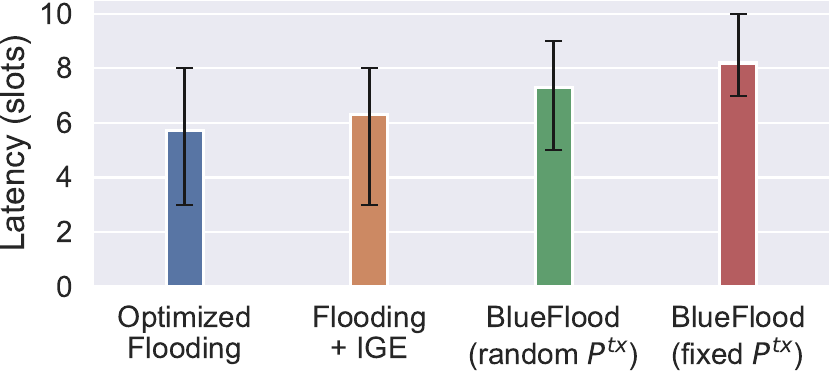}}
    \hspace{3em}
    \subfigure[Node coverage over time] 
    {\label{fig:metric4}
            \includegraphics[scale=0.4]{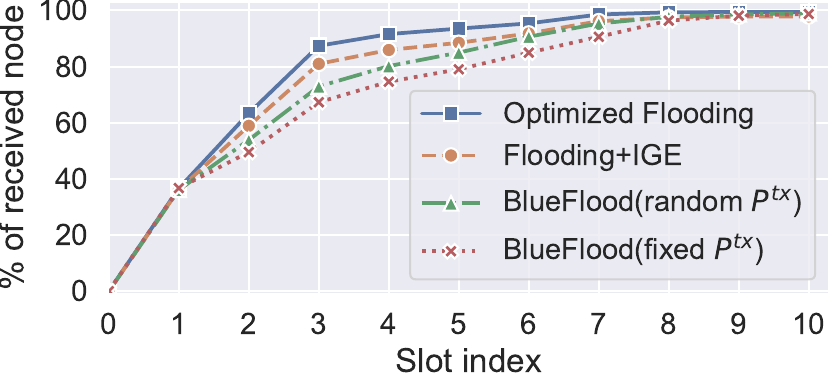}}
    \subfigure[Per-slot PER at different hops] 
    {\label{fig:metric5}
            \includegraphics[scale=0.4]{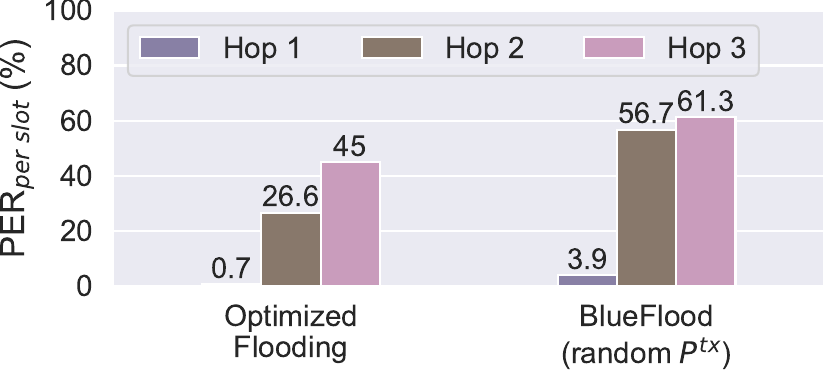}}
    \caption{Performance improvement of flooding}
    \Description{}{}
    \label{fig:performance_flooding}
\end{figure*}

\subsection{Performance improvement of flooding}

To understand how IGE helps improve the performance of flooding, we want to compare our approach with BlueFlood. As our approach employs variable transmit powers and BlueFlood uses fixed transmit powers, it is difficult to make a fair comparison between them. Instead, we allow BlueFlood to use a higher average transmit power, equating to have a better link quality. Besides, we also develop a variant of BlueFlood with random transmit powers for comparison. In our experiments, we compare the performance of four schemes: 1) Optimized Flooding, optimizing the transmit powers of nodes in flooding with the estimated interference graph in the flooding-alone period; 2) Flooding+IGE, simultaneous flooding and IGE in the IGE process; 3) random-power BlueFlood, a variant of BlueFlood with random transmit powers; 4) fixed-power BlueFlood, the original lower-power flooding protocol over BLE. We run the experiments of flooding for more than 3,000 rounds for each of the four schemes.
Figure \ref{fig:metric1} shows the per-slot PERs under the four schemes, where the packet size is of 255 bytes. Due to the high channel similarity in our testbed (Figure \ref{fig:hdiff}) and the large packet size, all the four schemes experience high PERs. In random-power BlueFlood, the channel similarity issue is mitigated using random transmit power such that it is easier for each node to have a channel of dominant received power. This increases the packet delievery rate of random-power BlueFlood,  outperforming the fixed-power BlueFlood. The rest two schemes, Optimized Flooding and Flooding+IGE, leverage the estimated interference graph to optimize the transmit powers of nodes for flooding and thus achieve lower per-slot PERs. Flooding+IGE performs worse than Optimized Flooding because it sacrifices performance for the purpose of IGE. As a result of the lower per-slow PER, Figure \ref{fig:metric2} shows that Optimized Flooding can achieve an end-to-end PER (E2E PER) of 2.4\%, where the E2E PER is the ratio between the count of failed flooding rounds and the total round count. In contrast, the high per-slot PERs of the random- and fixed-power BlueFloods cause them to have E2E PERs as high as 8.1\% and 13.4\%, respectively. 

Figure \ref{fig:metric3} shows the average latency for the four schemes in the unit of slots, where the latency is the time taken for all the nodes to receive the flooded packet and the error bars indicate the 10-th and 90-th percentile latencies, respectively. We can see that Optimized Flooding achieves the least latency among all the schemes and that it takes two fewer slots than the fixed-power BlueFlood on average to complete flooding. Meanwhile, even the 90-th percentile latency of Optimized Flooding is no larger than the average latency of the fixed-power BlueFlood. Moreover, simultaneous flooding and IGE only slightly increases the average latency compared with Optimized Flooding. Looking into how nodes are covered over time during flooding (Figure \ref{fig:metric4}), we find that on average Optimized Flooding can cover almost 90\% of nodes in the first three slots, greatly improving the likelihood and speed of full coverage. Due to the optimized transmit powers of nodes in Optimized Flooding, nodes at each hop have a lower per-slot PER than other schemes. This explains why Optimized Flooding is more reliable and faster.

\subsection{Power Consumption Overhead}
Since IGE requires RSSI sampling and transmit power adjustment, we want to understand the power consumption overhead of these operations. We first set the experiment node in receive mode as in the Flooding+IGE scheme, which will use bit counter to collect RSSI samples for every byte received. We measure the working voltage and current of the node for at least 1,000 rounds with Keysight E36311A Power Supply on three devices: nRF52832 dongle, nRF52832 development kit (dev kit) and nRF52840 dev kit. For fair comparison, we repeat the experiment with BlueFlood (fixed power) on each device.
Table \ref{tab:power_consumption} demonstrates the normalized average power consumption of BlueFlood and our IGE. Compared with BlueFlood, our IGE incurs only 1.60\% power overhead on the dongle and no more than 3\% on both development kits. To understand the overhead of transmit power adjustment, we enforced the nRF52 series devices to switch between different transmit powers for ten million times and found that the process only took 3.75 seconds. This equates to 375 nanoseconds per transmit power adjustment, which is three orders less than the slot time.
We can thus safely neglect the power and time overhead caused by the transmit power switching.

\begin{table}
\centering
  \caption{Normalized power consumption on average}
  \label{tab:power_consumption}
    \resizebox{0.5\columnwidth}{!}{%
    \begin{tabular}{c|c|c|c}
    \toprule
    Board & \multicolumn{1}{c|}{nRF52832 Dongle} & \multicolumn{1}{c|}{nRF52832 Dev Kit} & \multicolumn{1}{c}{nRF52840 Dev Kit} \\
    \midrule
    BlueFlood & 100\% & 100\% & 100\% \\
    Our IGE & 101.60\% & 102.28\% & 102.73\% \\
    \bottomrule
    \end{tabular}
    }
\end{table}

% \subsection{Success Probability of Power-Domain IGE} 

% {\color{red} Explain the related Figure \ref{fig:prob1}. Illustrate the failure probability by numbers (?)}

\section{Applications}
\label{sec:application}

To demonstrate the broader impact of IGE, we present two use cases for 1) convergecast and 2) channel map convergence  in BLE networks. We consider a multi-hop BLE network with diverse network activities, including divergecast, convergecast, and distributed peer-to-peer (P2P) communication. The divergecast is implemented through concurrent flooding such that power-domain IGE can be conducted. It has been shown above that IGE helps improve the performance of flooding-based divergecast. In this section, we focus on using the interference graphs estimated from IGE to improve convergecast and distributed P2P communication in BLE networks.

\subsection{Convergecast}
\noindent\textbf{Traditional Convergecast.} Convergecast is a common network activity for data collection. We consider a typical tree topology for convergecast, where each parent node manages one or more child nodes and iterates through its child nodes to collect data. The data is collected hop by hop from the farthest leaf nodes to the root node and each node transmits or receives data at scheduled time slots. Except for the root node, each non-leaf node acts as both a master and a slave. Parent and child nodes use adaptive frequency hopping (AFH) to mitigate external interference. Traditional AFH algorithms randomly select available channels from the channel map for frequency hopping~\cite{ble_spec}. Packet loss may occur when the channels of concurrent transmissions collide.

\begin{figure*}[t!]
\centering
    \subfigure[Convergecast success probabilities]
    {\label{fig:convergecast1}
            \includegraphics[scale=0.4]{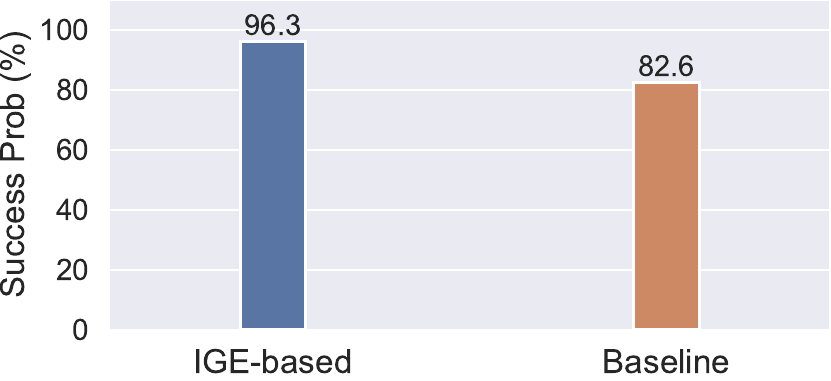}}
    \hspace{3em}
    \subfigure[Packet loss in failed convergecasts]
    {\label{fig:convergecast2}
            \includegraphics[scale=0.4]{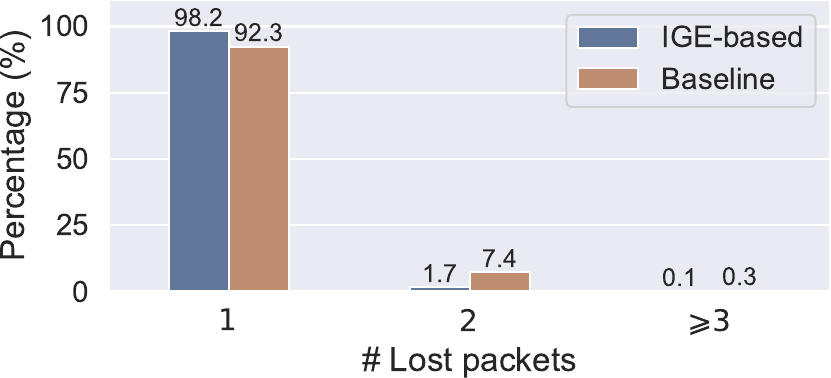}}
    \caption{Performance improvement of Convergecast}
    \Description{}{}
    \label{fig:performance_convergecast}
\end{figure*}

\noindent\textbf{IGE-based Convergecast.}  IGE enables us to estimate the additive interference from multiple concurrent transmissions on the same channel. Unlike traditional convergecast, IGE-based convergecast can identify links that can transmit concurrently on the same channel and those that cannot. We can avoid packet loss due to co-channel interference by allocating non-overlapping channels to mutually interfering links. To ease channel allocation, we put non-mutually interfering links in the same group and assign them the same set of channels. More specifically, we group links in a greedy manner: we first find the largest group of non-mutually interfering links from the ungrouped links and repeat the process until every link is in a group.

\noindent\textbf{Experimental Settings and Results.}
We construct a tree topology connecting the root node and the rest of nodes based on the full network topology in Figure~\ref{fig:topology}. Except for the leaf nodes, nodes with even indices are selected to be parent nodes. Specifically, node $0$ manages all nodes in Hop 1, and each node with an even index in Hops 1 and 2 has two child nodes, i.e., node $i$ is the parent of nodes $i+5$ and $i+6$. We consider the traditional convergecast as the baseline, where parent nodes randomly pick channels from the same set of available channels for channel hopping. For both IGE-based convergecast and baseline, no packet is retransmitted when packet loss occurs.
Figure ~\ref{fig:convergecast1} shows the success probabilities of the two schemes for 10,000 repeated experiments, where each node has 37 available channels (namely all the BLE data transmission channels). It can be seen that IGE-based convergecast can effectively mitigate the impact of co-channel interference and achieve a success probability of 96.3\%. This demonstrates that IGE succeeds in identifying mutually interfering links and allocate them non-overlapping channels. In contrast, the baseline only has a success probability of 82.6\% due to co-channel interference. For failed convergecasts, we look into the percentage of nodes whose data are successfully collected. As shown in Figure ~\ref{fig:convergecast2}, IGE-based convergecast can collect data from more nodes in the network than the baseline. Leveraging IGE can help convergecast to avoid co-channel interference and reduce packet loss.

\subsection{Channel Map Convergence}

\noindent\textbf{Traditional vs IGE-assisted P2P Communication.} In BLE networks, distributed P2P communications leverage frequency hopping to mitigate the impact of interference. For each transmitter-receiver pair, the available channels for frequency hopping are selected based on the channel map, which is dynamically updated by continuously monitoring the interference. A recent work ~\cite{pdr_exclusion} proposes to use the packet loss rate on each channel to update the channel map, where a channel is discarded if the packet loss rate on that channel exceeds a threshold. When multiple interferers exist, it generally takes long time for the channel maps of nodes to converge and packet loss occurs during the convergence process. With IGE, we can allocate non-overlapping channels to interfering links to both accelerate the convergence and reduce packet loss during the process. More specifically, we put links into groups such that links in the same group can transmit concurrently. We group links in a greedy manner based on the estimated interference graph, where the SINRs of nodes in each cluster exceed a threshold. We greedily find the maximum group for ungrouped links and repeat this for the rest of ungrouped links until each link belongs to a group. For fairness, we evenly divide total available channels between groups.

\begin{figure*}[t!]
\centering
    \subfigure[PDR of each transmitter-receiver pair]
    {\label{fig:channelmap1}
            \includegraphics[scale=0.4]{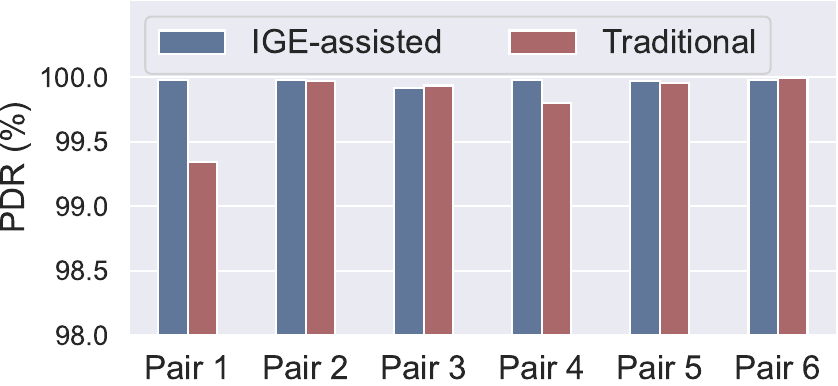}}
    \hspace{3em}
    \subfigure[Cumulative number of lost packets]
    {\label{fig:channelmap2}
            \includegraphics[scale=0.4]{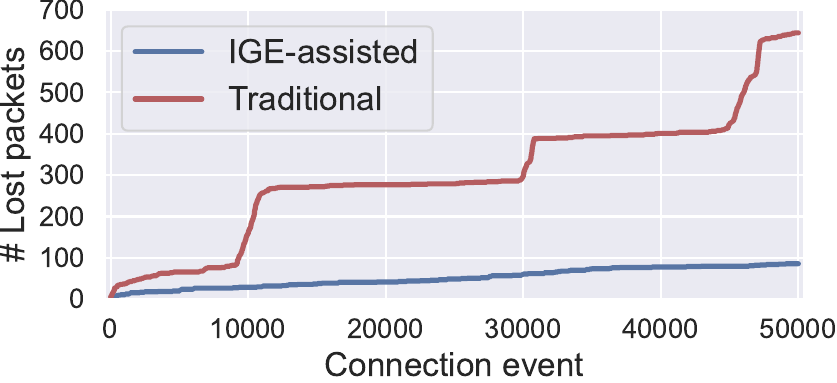}}
    \caption{Performance improvement of channel map convergence}
    \Description{}{}
    \label{fig:performance_channelmap}
\end{figure*}

\noindent\textbf{Experimental Settings and Results.}
In our experiments, the initial channel map for each node includes 12 available channels. Traditional P2P directly employs AFH using the initial channel map, while the IGE-assisted P2P adjusts the channel maps of nodes based on the interference graph and employs AFH for data transmission using the adjusted channel maps. Figure ~\ref{fig:channelmap1} compares the packet delivery rate (PDR) between traditional and IGE-based P2Ps during the channel map convergence, where 6 transmitter-receiver pairs are formed with nodes at the 2nd and 3rd hops in our testbed. It can be seen that IGE-based P2P is more reliable than the traditional P2P and achieves a PDR of 99.9\% even for the worst pair. Although traditional P2P also achieves a high average PDR, it suffers from severe sudden packet drops. As shown in Figure~\ref{fig:channelmap2}, there are three sudden jumps of packet loss, during which the packet loss rates are 8.7\%, 11.3\%, and 8.8\%, respectively. This is because pairs 1 and 4 detect high packet loss rates on most available channels and have to reset their channel maps for a new round of convergence.

\section{Discussions \& Limitations}
\noindent\textbf{Nonlinearity Issues.}
We have identified four nonlinearity issues found in our experiments and proposed several methods to control almost all the power ratios to fall between 0.9 and 1.1. However, there are still some outliers beyond this range. We ensure that these outliers are not measurement artifacts by excluding all possible known causes. Nonetheless, the limited visibility into the hardware of nRF52 series SoCs prevents us from diving deeper into this issue and further pushing the power ratios towards 1. The nonlinearity issues weaken the linear relation between the transmit and received powers and affect the accuracy of IGE. Moreover, we limit the transmit powers of nodes within the linear region below 0dBm. Using a higher transmit power would greatly improve the communication range. Based on the path loss exponent model, when the path-loss exponent is 2, using the maximum transmit power (8dBm) of nRF52 series SoCs may improve the communication distance by 2.5 times than using a transmit power of 0dBm.

\noindent\textbf{Scalability Issues.} IGE requires constructing a full-rank transmit power matrix for nodes in the same hop. When the number of same-hop nodes is large, it takes a long time to construct the transmit power matrix. A longer IGE process does not delay flooding as flooding is conducted simultaneously with IGE, but it slows down the update frequency of the interference graph, making it unable to follow the channel variations. For fast varying channels, the use of IGE is limited to small-to-medium-sized networks.

\noindent\textbf{Centralized Control, Overhead of IGE, and Optimizations.} This work assumes a network with the needs of divergecast and convergecast, where a central node manages the network by collecting measurement data and disseminating decision plans. The integration of IGE introduces extra but small communication overhead to the network as analyzed in \S\ref{sec:reporting_overhead}. We can reduce the overhead of IGE by 1) reducing the frequency of IGE or only conducting IGE when needed, 2) designing an efficient representation for control and measurement data, and 3) piggybacking these data onto normal traffic. Besides, estimating interference graphs also consumes computing resources. We have done this with a laptop connected to the central node, which is not necessary if efficient algorithms can be designed for the low-power devices. We may also consider decentralized power control for future optimizations, which alleviates both the communication and computing workload for the central node and is robust to the single point of failure.

\noindent\textbf{Broader Impact of the Interference Graph.}
We have demonstrated how to achieve IGE together with CT-based flooding and improve the performance of flooding with the estimated interference graph. In fact, as interference graphs are widely used in wireless network resource management, the interference graph estimated from flooding can also be used to benefit many other network activities. Considering the multi-faceted contributions of IGE to the network, the related overhead is worthwhile.

\section{Conclusions}
In this paper, we proposed to integrate interference graph estimation into data transmission tasks such that it can be done simultaneously with the data transmission tasks using the same frequency-time resources. We found that concurrent flooding is a perfect carrier for this integration in wireless sensor networks and designed network protocols and power allocation algorithms for efficient interference graph estimation. We showed that interference graph estimation on the Nordic nRF52 series SoCs was feasible in both controlled and real-world experiments and that flooding with power control based on the estimated interference graph outperformed the plain flooding. In addition to enhancing flooding, interference graph estimation was demonstrated to have potential for a broader range of tasks in BLE networks.

\nocite{ourEWSN} 
\bibliographystyle{ACM-Reference-Format}
\bibliography{sample-base}

\end{document}